\documentclass[%
preprint,
nofootinbib,
 amsmath,amssymb,
 aps,
pra,
floatfix,
]{revtex4-2}

\usepackage{graphicx}% Include figure files
\usepackage{dcolumn}% Align table columns on decimal point
\usepackage{bm}% bold math
\usepackage{color}

\usepackage{tikz}
\usetikzlibrary{quantikz2}
\usetikzlibrary{snakes}

\usepackage[utf8]{inputenc}
\usepackage[all]{xy} \xyoption{pdf}
\usepackage{braket}
\usepackage{amsthm}
\newtheorem{dfn}{Definition}[section]
\newtheorem{prop}[dfn]{Proposition}
\newtheorem*{prop*}{Proposition}

\newtheorem{thm}[dfn]{Theorem}

\newtheorem{rem*}[dfn]{Remark}

\newcommand{\CNOT}{\mathrm{CNOT}}
\newcommand{\CZladder}{Ladder_\text{CZ}}
\newcommand{\CNOTladder}{Ladder_\CNOT}
\newcommand{\CNOTladderN}[1]{Ladder_{\CNOT,#1}}
\newcommand{\RyRz}{R_y\text{-}R_z}
\newcommand{\Ry}{R_y}
\newcommand{\Rz}{R_z}
\newcommand{\HEA}{\mathrm{HEA}}
\newcommand{\Uent}{U_{\text{ent}}}

\newcommand{\HEARyRzCZ}{\HEA(R_y\text{-}R_z,\CZladder)}
\newcommand{\HEARyCNOT}{\HEA(R_y,\CNOTladder)}

\newcommand{\poly}{\mathrm{poly}}

\begin{document}

% \preprint{APS/123-QED}

\title{On universality of hardware-efficient ansatzes}% Force line breaks with \\

\author{Hokuto Iwakiri}
 \email{iwakiri@qunasys.com}
\author{Keita Kanno}%
 \email{kanno@qunasys.com}
\affiliation{%
 QunaSys
}%

\date{\today}% It is always \today, today,

\begin{abstract}
The hardware-efficient ansatz (HEA) is one of the most important class of parametrized quantum circuits for near-term applications of quantum computing. We show that the problem of simulating some major classes of the HEA is BQP-complete by explicitly demonstrating that any relevant quantum circuit can be efficiently represented as an HEA circuit of those classes.
\end{abstract}

\maketitle
\newpage % to put the abstract on the first page
%\tableofcontents
\section{Introduction}

The hardware-efficient ansatz (HEA)~\cite{Hardware-Efficient} is one of the most widely used class of parametrized quantum circuits for near-term applications of quantum computing. The circuit consists of layers of parametrized single-qubit gates and constant entangling two-qubit gates, often configured according to the hardware connectivity.

HEA is widely used in near-term quantum algorithms such as variational quantum eigensolver (VQE) \cite{VQE} and quantum machine learning algorithms \cite{1802.06002, 1803.00745}, 
as it takes advantage of the hardware connectivity for achieving maximal representability with limited number of quantum gates.

As it is an essential component of near-term quantum computing, numerous studies have also been devoted to investigate ways to classically simulate HEA circuits \cite{1910.09534, 2002.01935,2005.06787, 2211.03999,2409.01706, mitarai2022quadratic}.
In particular, Ref.~\cite{2409.01706} shows that generic\footnote{i.e., except for some special cases} quantum circuits in a HEA-like family is efficiently simulable using Pauli propagation methods.
It is thus of great interest to rigorously investigate whether the classical simulation of HEA is possible in any cases, and we will prove that the answer is no.

Hardness of classical simulation has been proved for various types of quantum circuits.
An early example is the IQP circuits, consisting only of gates commuting with Hadamard operations, yet shown to yield output distributions hard to sample classically~\cite{doi:10.1098/rspa.2010.0301}. For practically motivated circuit families, similar sampling-hardness result has been established~\cite{farhi2019quantumsupremacyquantumapproximate} for quantum approximate optimization algorithm (QAOA)~\cite{1411.4028}, an algorithm for combinatorial optimization. Moreover, in the context of chemistry-inspired ansatzes, there exist results showing hardness for fermionic systems that incorporate first-excitation operators~\cite{jozsa2008matchgates, brod2011extending, brod2012geometries,brod2016efficient, hebenstreit2019all,beenakker2004charge,arrazola2022universal,divincenzo2005fermionic,oszmaniec2017universal} as well as for ansatzes including higher-order excitations such as the unitary coupled-cluster with singles and doubles (UCCSD) circuits~\cite{1701.02691, 2503.21041} and unitary cluster Jastrow circuits~\cite{doi:10.1021/acs.jctc.9b00963, 2504.12893}.

In this study, we will show, explicitly by construction, that simulating a class of HEA is BQP-complete.
This implies that simulating, not only sampling from but also computing expectation values of, such HEA is classically hard unless BQP $\subset$ P.

Before describing the specific architectures considered in this work, we remark that circuits closely related to the HEA have been shown to be universal even under strong structural constraints.
Raussendorf~\cite{Raussendorf2005TI} demonstrated that a translation-invariant circuit consisting of alternating layers of uniform single-qubit rotations, all applied with the same rotation angle, and nearest-neighbor CZ interactions is universal for quantum computation.
In this work, we extend the scope by considering HEA families that allow independently parameterized Ry and Rz rotations, yet exclude Rx rotations, and we prove that such HEA families are BQP-complete.

More specifically, we consider two types of HEA: (i) one employing Ry and Rz rotations together with a linear CZ ladder, and (ii) one employing only Ry rotations together with a linear CNOT ladder. We will show that simulating both types of HEAs are BQP-complete. Note that the hardness result of the latter HEA immediately implies the hardness of simulating HEAs with Ry and Rz rotations and linear CNOTs. 

To show the BQP-completeness, we first show that all the gates of a universal gate set can be constructed with a HEA with depth that is polynomial in the system size, which implies that simulating HEA is BQP-hard. It is obvious that simulating polynomial-depth HEA is contained in BQP, so we can conclude that simulating such HEA is BQP-complete.

The rest of the paper is structured as follows: in Section~\ref{preliminaries}, we introduce essential ingredients to show the main result of this paper. We formally define HEA, and introduce the notion of universality for circuit families. Section~\ref{Results} describes the main results of the paper and their proofs. In Section~\ref{sec:discussion} we summarize the result and discuss its implications and potential future directions, and finally conclude in Section~\ref{sec:conclusion}. Proofs of some results used in the proof of the main result is given in the appendix.

\section{Preliminaries}\label{preliminaries}
In this section, immediately after reviewing the general form of the hardware-efficient ansatz (HEA)~\cite{Hardware-Efficient}, we introduce the two universality notions that will underpin our proof of HEA universality. In Section \ref{subsec:Definition of HEA}, a general HEA on \(N\) qubits consists of \(M\) alternating layers of single-qubit rotations \(R\) on every wire and a fixed two-qubit entangler \(\Uent\). 
Later in this paper, we will show that, when \(R,\Uent\) are chosen as either a $\RyRz$ and a CNOT ladder or a $\Ry$ and a CZ ladder, the corresponding HEA remains universal.
To make the universality claim more concrete, in Section \ref{subsec:universal_gate_sets}, we first introduce \emph{strict universality} and \emph{computational universality} for quantum gate sets, and then define the notion of universality for quantum circuit families.

\subsection{Definition of HEA}\label{subsec:Definition of HEA}
In this section, we firstly define the hardware-efficient ansatz~\cite{Hardware-Efficient} in a general form, and then introduce some restricted versions of HEA, which we will show to be universal.

\begin{figure}[htbp]
    \centering
    \makebox[\columnwidth][c]{ 
    \begin{quantikz}
    \lstick{$\ket{q_N}$}     & \gate{R} & \gate[wires=5]{U_{\text{ent}}} & \gate{R} & \gate[wires=5]{U_{\text{ent}}} & & \cdots & \gate{R} & \\
    \lstick{$\ket{q_{N-1}}$} & \gate{R} &                                 & \gate{R} &                             & &    \cdots    & \gate{R} & \\
    \lstick{$\vdots$} \setwiretype{n}       & \vdots     &                                 & \vdots     &                             &     &        & \vdots     &      \\
    \lstick{$\ket{q_{2}}$}   & \gate{R} &                                 & \gate{R} &                             & &    \cdots    & \gate{R} & \\
    \lstick{$\ket{q_{1}}$}   & \gate{R} &                                 & \gate{R} &                             & &   \cdots     & \gate{R} & \qw
    \end{quantikz}
    }
    \caption{The general structure of hardware-efficient ansatz.}
    \label{fig:Original_HEA}
\end{figure}
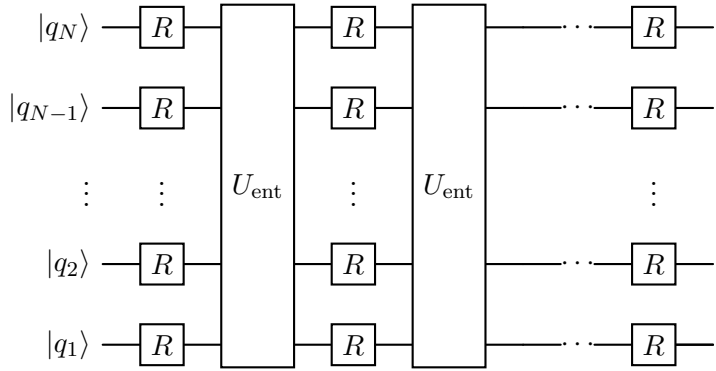

\begin{figure}[tb]
    \makebox[\columnwidth][c]{
\begin{quantikz}[row sep=0.3em, column sep=0.3em]
\lstick{$\ket{q_{N}}$}     & \gate{\Ry} & \gate{\Rz} & \ctrl{1} & & & \gate{\Ry} & \gate{\Rz} & & \ctrl{1} & & & \gate{\Ry} & \gate{\Rz} & & \ctrl{1} & & & \gate{\Ry} & \gate{\Rz} & & \ctrl{1} & & & \cdots \\
\lstick{$\ket{q_{N-1}}$}   & \gate{\Ry} & \gate{\Rz} & \ctrl{0} & \ctrl{1} & & \gate{\Ry} & \gate{\Rz} & & \ctrl{0} & \ctrl{1} & & \gate{\Ry} & \gate{\Rz} & & \ctrl{0} & \ctrl{1} & & \gate{\Ry} & \gate{\Rz} & & \ctrl{0} & \ctrl{1} & & \cdots \\
\lstick{$\vdots$} \setwiretype{n}          &            &            &          & \vdots   &      &            &            &      &          & \vdots   &      &            &            &      & \vdots   &      &      &            &            &      & \vdots   &      &      &        \\
\lstick{$\ket{q_{2}}$}     & \gate{\Ry} & \gate{\Rz} & \ctrl{1} & \ctrl{0} & & \gate{\Ry} & \gate{\Rz} & & \ctrl{1} & \ctrl{0} & & \gate{\Ry} & \gate{\Rz} & & \ctrl{1} & \ctrl{0} & & \gate{\Ry} & \gate{\Rz} & & \ctrl{1} & \ctrl{0} & & \cdots \\
\lstick{$\ket{q_{1}}$}     & \gate{\Ry} & \gate{\Rz} & \ctrl{0} & & & \gate{\Ry} & \gate{\Rz} & & \ctrl{0} & & & \gate{\Ry} & \gate{\Rz} & & \ctrl{0} & & & \gate{\Ry} & \gate{\Rz} & & \ctrl{0} & & & \cdots
\end{quantikz}

    }
    \caption{The Ry-Rz-CZ ansatz.}
    \label{fig:Ry_Rz}
\end{figure}
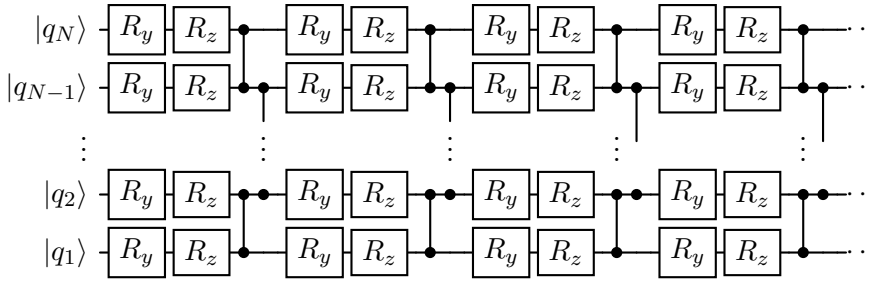

\begin{figure}[tb]
    \centering
    \makebox[\textwidth]{
\begin{quantikz}
\lstick{$\ket{q_{N}}$}     & & \ctrl{1} &      & \\
\lstick{$\ket{q_{N-1}}$}   & & \ctrl{0} & \ctrl{1} & \\
\lstick{$\ket{q_{N-2}}$}   & & \ctrl{1} & \ctrl{0} & \\
\lstick{$\ket{q_{N-3}}$}   & & \ctrl{0} & \ctrl{1} & \\
\lstick{$\vdots$}          \setwiretype{n}&     &          & \vdots   &      \\
\lstick{$\ket{q_{2}}$}     & & \ctrl{1} & \ctrl{0} & \\
\lstick{$\ket{q_{1}}$}     & & \ctrl{0} &      &
\end{quantikz}
    }
    \caption{$N$-qubit ladder of CZ. Since CZs are mutually commutative, unlike the $CNOT$ ladder, the structure does not change whether even or odd.}
    \label{fig:CZ_ladder}
\end{figure}
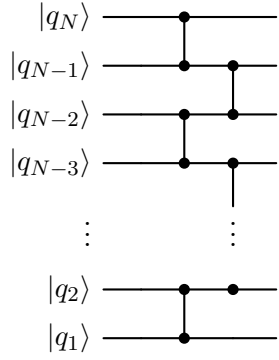

Throughout this paper, we define a general HEA to be a quantum circuit of the following structure (Fig. \ref{fig:Original_HEA})~\cite{Hardware-Efficient}
\begin{equation}
    \HEA(R, \Uent, M)=(\prod_{i=1}^M U_1^{(i)} \Uent)U_1^{(0)}.
\end{equation}
Here, $R$ is some parametrized single-qubit gate,  $\Uent$ is a set of constant two-qubit gates, $M$ is a non-negative integer which we call \emph{depth} of the ansatz, and $U_1^{(i)}$, for each $i$, is a quantum circuit composed of $R$ acting on every qubit. Each parameter for the $R$'s can be different from each other.

In this paper, we will show that two major types of HEA are universal. The first is the \emph{Ry-Rz-CZ ansatz} (Fig. \ref{fig:Ry_Rz}),
\begin{align}
\HEA(\RyRz, \CZladder, M)=(\prod_{i=1}^M (\RyRz)^{(i)} \CZladder)(\RyRz)^{(0)},
\end{align}
where $\RyRz$ denotes a sequence of $\Ry$ and $\Rz$ gates, and $\CZladder$ is a linear ladder of $\text{CZ}$ gate shown in Fig.~\ref{fig:CZ_ladder}. The other is the \emph{Ry-CNOT ansatz} (Fig.~\ref{fig:Ry}),
\begin{align}
\HEA(\Ry, \CNOTladder, M)=(\prod_{i=1}^M (\Ry)^{(i)} \CNOTladder)(\Ry)^{(0)},
\end{align}
where $\Ry$ is a sequence of $\Ry$ gate, and $\CNOTladder$ is a linear ladder of $\CNOT$ gate shown in Fig.~\ref{fig:CNOT_ladder}.

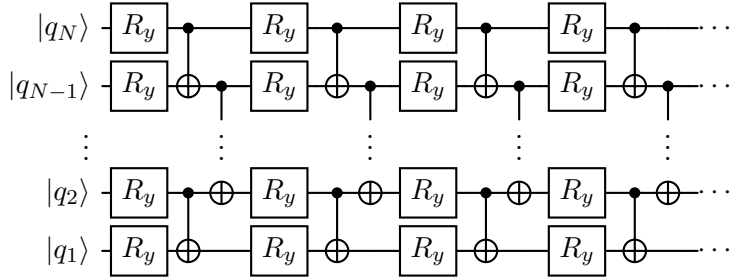
\begin{figure}[tb]
    \makebox[\columnwidth][c]{
\begin{quantikz}[row sep=0.3em, column sep=0.3em]
\lstick{$\ket{q_{N}}$}     & \gate{\Ry} & \ctrl{1} &      & & \gate{\Ry} & & \ctrl{1} &      & & \gate{\Ry} & & \ctrl{1} &      & & \gate{\Ry} & & \ctrl{1} &      & & \cdots \\
\lstick{$\ket{q_{N-1}}$}   & \gate{\Ry} & \targ{}  & \ctrl{1} & & \gate{\Ry} & & \targ{}  & \ctrl{1} & & \gate{\Ry} & & \targ{}  & \ctrl{1} & & \gate{\Ry} & & \targ{}  & \ctrl{1} & & \cdots \\
\lstick{$\vdots$} \setwiretype{n}          &            &          & \vdots   &     &            &     &          & \vdots   &     &            &     &          & \vdots   &     &            &     &          & \vdots   &     &        \\
\lstick{$\ket{q_{2}}$}     & \gate{\Ry} & \ctrl{1} & \targ{}  & & \gate{\Ry} & & \ctrl{1} & \targ{}  & & \gate{\Ry} & & \ctrl{1} & \targ{}  & & \gate{\Ry} & & \ctrl{1} & \targ{}  & & \cdots \\
\lstick{$\ket{q_{1}}$}     & \gate{\Ry} & \targ{}  &      & & \gate{\Ry} & & \targ{}  &      & & \gate{\Ry} & & \targ{}  &      & & \gate{\Ry} & & \targ{}  &      & & \cdots
\end{quantikz}

    }
    \caption{The Ry-CNOT ansatz.}
    \label{fig:Ry}
\end{figure}

\begin{figure}[tb]
    \centering
    % \resizebox{8.3cm}{!}{
    \makebox[\textwidth]{
    \begin{quantikz}
    \lstick{$\ket{q_{2i}}$}     & & \ctrl{1} &      & \\
    \lstick{$\ket{q_{2i-1}}$}   & & \targ{}  & \ctrl{1} & \\
    \lstick{$\ket{q_{2i-2}}$}   & & \ctrl{1} & \targ{}  & \\
    \lstick{$\ket{q_{2i-3}}$}   & & \targ{}  & \ctrl{1} & \\
    \lstick{$\vdots$}           \setwiretype{n}&     &          & \vdots   &      \\
    \lstick{$\ket{q_{2}}$}      & & \ctrl{1} & \targ{}  & \\
    \lstick{$\ket{q_{1}}$}      & & \targ{}  &      & \qw
    \end{quantikz}
    \quad
    \quad
    \quad
    \quad
\begin{quantikz}
\lstick{$\ket{q_{2i+1}}$}   & & \ctrl{1} &      & \\
\lstick{$\ket{q_{2i}}$}     & & \targ{}  & \ctrl{1} & \\
\lstick{$\ket{q_{2i-1}}$}   & & \ctrl{1} & \targ{}  & \\
\lstick{$\vdots$}      \setwiretype{n}     &     & \vdots   &          &      \\
\lstick{$\ket{q_{2}}$}      & & \targ{}  & \ctrl{1} & \\
\lstick{$\ket{q_{1}}$}      & &      & \targ{}  & \qw
\end{quantikz}
    }
    % }
    \caption{Left side is a $(2i)$-qubit ladder of $CNOT$ and Right side is a $(2i+1)$-qubit ladder of CNOT.}
    \label{fig:CNOT_ladder}
\end{figure}
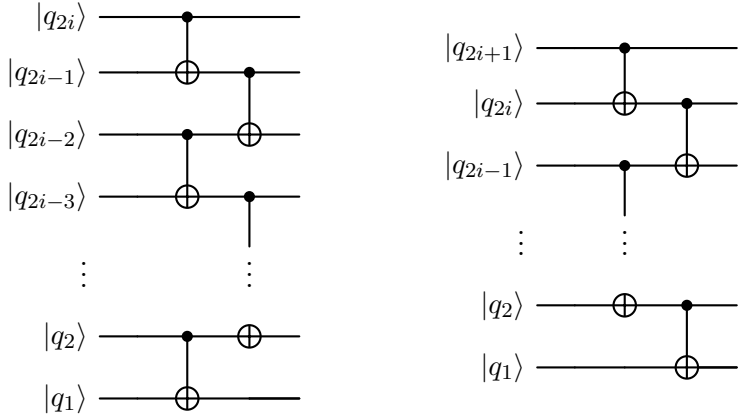

\subsection{Strict and computational universality}
\label{subsec:universal_gate_sets}
Now let us define some notion of universalities for families of quantum circuits, such as types of HEA. To this end, we first review a standard notion of universality for gate sets.

A gate set $\mathcal G$ is a finite collection of quantum gates.  An $n$-qubit circuit $C_n$ over $\mathcal G$ is any finite sequence of these gates; its size $|C_n|$ is the total gate count.
In what follows, we introduce two distinct notions of universality for a gate set $\mathcal G$.

\subsubsection{Strict universality}
A gate set \(\mathcal G\) is \emph{strictly universal} if there exists an integer \(n_{0}\) such that for every \(n\ge n_{0}\) the subgroup generated by \(\mathcal G\) is dense in \(SU(2^n)\)~\cite{aharonov2003}.  Equivalently, for any \(n\)-qubit circuit \(C_n\) over a strictly universal basis and any tolerance \(\varepsilon>0\), there is another \(n\)-qubit circuit \(C'_n\) using only gates from \(\mathcal G\) such that
\begin{align}
\bigl\|\,C_n - C'_n\bigr\|\;\le\;\varepsilon,
\end{align}
where  $\bigl\|\ \cdot \bigr\|$ is the standard operator norm
induced by the $l_2$ norm.  Although the definition does not itself control the growth of \(\lvert C'_n \rvert\) as \(\varepsilon\to0\), the Solovay–Kitaev theorem ensures that for each fixed \(n\ge n_{0}\) one can choose
\begin{align}
\lvert C'_n\rvert \;=\; O(\mathrm{polylog}(1/\varepsilon)),
\end{align}
and that such a circuit can be found by a classical algorithm in time $O(\mathrm{polylog}(1/\varepsilon))$~\cite{dawson_nielsen_2005}.

An example of strictly universal gate set is the Clifford–\(T\) set: Hadamard gate $H$, T-gate $T$, Controlled-Not $\mathrm{CNOT}$ since \(H,T\) densely generate \(SU(2)\) and CNOT provides the requisite entangling capability~\cite{nielsen2010}.

\subsubsection{Computational universality}\label{sec:computational_universality_gate_set}
A gate set \(\mathcal G\) is \emph{computationally universal} if, for any uniform family of polynomial-size circuits \(\{C_n\}\) whose gates are drawn from a strictly universal set, one can construct replacement circuits \(\{C_n'(\varepsilon)\}\) using only gates from \(\mathcal G\) such that the entire output probability distribution of each \(C_n\) is reproduced to within total-variation error \(\varepsilon\) for an error tolerance \(\varepsilon>0\), and moreover
\[
\lvert C_n'(\varepsilon)\rvert \;=\;\mathrm{poly}\!\bigl(\lvert C_n\rvert,\log\tfrac1\varepsilon\bigr).
\]
In this sense, \(\mathcal G\) can emulate the sampling behavior of every polynomial-size quantum algorithm originally built from a strictly universal basis, incurring only polynomial-factor overhead in circuit size and accuracy~\cite{aharonov2003,nielsen2010}. 

Shi~\cite{R_CNOT_universality} proved that
$\{R,\mathrm{CNOT}\}$ is \emph{computationally universal} where $R$ is any \emph{real} single-qubit
gate whose square is not diagonal in the computational basis. 
Product of $R$ rotations densely generate $SO(2)$, and the entangling gate propagates this real subgroup to $SO(2^{n})$ for all $n$.
Because only real matrices are produced, the generated group is \emph{not} dense in $SU(2^{n})$, so strict universality fails, but Shi’s poly-time construction ensures computational universality.

\subsubsection{Universality of families of quantum circuits}
In order to relate these notions of universality to HEA, we introduce the following notion of universality for families of quantum circuits. Let us consider a family of quantum circuits $\mathcal{C}$ with an integer parameter called \emph{depth}, which counts the number of blocks in each circuit, where each block constitutes of at most $\poly(N)$ elementary gates of some universal gate set.  Such family of quantum circuits $\mathcal{C}$ is \textit{strictly} (resp. \textit{computationally}) \textit{universal} if and only if, for some gate set $\mathcal{G}$ that is strictly (resp. computationally) universal, the following holds.
For any $N$-qubit quantum circuit, which we will call a \emph{target circuit} throughout the rest of the paper, that is made of gates in $\mathcal{G}$ that are polynomially many in $N$, there is a corresponding quantum circuit in $\mathcal{C}$ with $N+O(1)$ qubits and $\poly(N)$ depth which is exactly equivalent as a unitary operation (up to an addition of trivial ancilla qubits and global phases).
We will show that the Ry-Rz-CZ ansatz and the Ry-CNOT ansatz are universal in this sense.

\section{Results}\label{Results}
In this section we prove that the Ry-Rz-CZ ansatz is strictly universal and the Ry-CNOT ansatz computationally universal. To this end, we show that the following two claims hold: For two HEA circuits $C_1$ and $C_2$ of same type with depth $M_1$ and $M_2$, respectively, a concatenated circuit $C_1+C_2$ is equivalent to a HEA circuit with depth $M_1+M_2+M'$, where $M'$ is an integer which is constant independent of the number of qubits $N$. Then we show that for every gate, embedded in an $N$-qubit circuit, in a gate set that is strictly (resp. computationally) universal, there is an $N+O(1)$-qubit HEA representation of the gate, where the HEA circuit has depth at most $\poly(N)$. These two facts mean that one can construct an HEA circuit with depth at most $\poly(N)$ for any target circuit that is generated by polynomially many gates that is in the gate set, concluding our universality proof.

\begin{thm}\label{thm:cz_universal}
  The circuit family $\HEARyRzCZ$ is strictly universal.
\end{thm}

\begin{thm}\label{thm:cnot_universal}
  The circuit family $\HEARyCNOT$ is computationally universal.
\end{thm}

\subsection{Decomposition into a product of unitary layers}\label{sec:decomposition}

We now show that HEA circuits of our interest can be merged into one HEA circuit with at most constant-depth overhead. Namely, a sequence of two Ry-Rz-CZ circuits of depth $M_1$ and $M_2$ is equivalent to a Ry-Rz-CZ circuit of depth $M_1+M_2+2$, and a sequence of two Ry-CNOT circuits of depth $M_1$ and $M_2$ is equivalent to a Ry-CNOT circuit of depth $M_1+M_2$, i.e., no additional depth is required.

\subsubsection{$\HEA(R_y\text{-}R_z,\CZladder)$}\label{sec:cz_decomp}
Let us first consider the Ry-Rz-CZ ansatz.
Consider to merge two Ry-Rz-CZ circuits, $C_1$ and $C_2$ of depth $M_1$ and $M_2$, respectively  into one Ry-Rz-CZ circuit $C$ (Fig.~\ref{fig:concatenation-ry-rz-cz}). If one naively concatenate the two circuits, then there are sequences of Ry-Rz-Ry-Rz at the concatenation point, and is not a valid Ry-Rz-CZ circuit. However, one can consider depth-$M_1+M_2+2$ Ry-Rz-CZ circuit, with the rotation angles of the middle rotation gates set to be zero. Now one can see that a concatenation can be done with two additional depth.
It is thus in general true that a concatenation of $M_1, M_2, \dots, M_d$-depth Ry-Rz-CZ circuits can be represented as one Ry-Rz-CZ circuit of depth $M$,
\begin{equation}\label{eq:depth-RyRz}
    M=2d-2+\sum_{j=1}^d M_j.
\end{equation}
\begin{figure}
    \centering
    \begin{tikzpicture}
		\node[scale=0.7] at (0,0) {
    \begin{quantikz}[row sep={25pt,between origins}, column sep=10pt]
         \ldots \ & \ctrl{1} &       & \gate{R_y} & \gate{R_z} &\\
          \ldots \  & \control{} & \ctrl{1} & \gate{R_y} & \gate{R_z}&\\
         \ldots \ &           & \control{} & \gate{R_y} & \gate{R_z}&
    \end{quantikz}
    +
    \begin{quantikz}[row sep={25pt,between origins}, column sep=10pt]
         & \gate{R_y} & \gate{R_z} & \ctrl{1} &       & \ \ldots   \\
         & \gate{R_y} & \gate{R_z} & \control{} & \ctrl{1} & \ \ldots \\
         & \gate{R_y} & \gate{R_z} &           & \control{} & \ \ldots 
    \end{quantikz}
    =
    \begin{quantikz}[row sep={25pt,between origins}, column sep=10pt]
         \ldots \ & \ctrl{1}	&			& \gate{R_y} &\gate{R_z}	& \ctrl{1}	&			& \ctrl{1}	&			&\gate{R_y} &\gate{R_z} & \ctrl{1}	&			&  \ \ldots  \\
          \ldots \ & \control{}	& \ctrl{1}	& \gate{R_y} &\gate{R_z}	& \control{}	& \ctrl{1}	& \control{}	& \ctrl{1}	&\gate{R_y} &\gate{R_z} & \control{}	& \ctrl{1}	& \ \ldots  \\
           \ldots \ & 			& \control{}	& \gate{R_y} &\gate{R_z}	 & 			& \control{}	& 			& \control{}	&\gate{R_y} &\gate{R_z} &			& \control{}	& \ \ldots 
    \end{quantikz}
    };
    \end{tikzpicture}
    \caption{Concatenation of Ry-Rz-CZ circuits}
    \label{fig:concatenation-ry-rz-cz}
\end{figure}
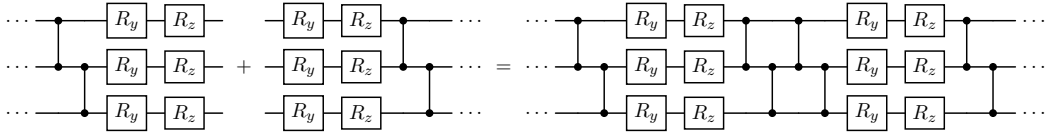

\subsubsection{$\HEA(R_y,\CNOTladder)$}\label{sec:cnot_decomp}
We now consider the Ry-CNOT ansatz. The same logic as in the previous case does not apply, as $\CNOTladder^2$ does not equal to identity. But actually the concatenation is easier, and can be done without any additional depth in this case. Consider concatenating two Ry-CNOT ansatz of depth $M_1$ and $M_2$ (Fig.~\ref{fig:concatenation-ry-cnot}). There are sequence of two Ry gates at the concatenation point, but as $R_y(\alpha) R_y(\beta)=R_y(\alpha+\beta)$, the concatenated circuit is of the form of depth-$(M_1+M_2)$ Ry-CNOT ansatz. A concatenation of $M_1,M_2,\dots,M_d$-depth Ry-CNOT circuits can be represented as one Ry-CNOT circuit of depth $M$,
\begin{equation}\label{eq:depth-Ry}
    M=\sum_{j=1}^d M_j.    
\end{equation}
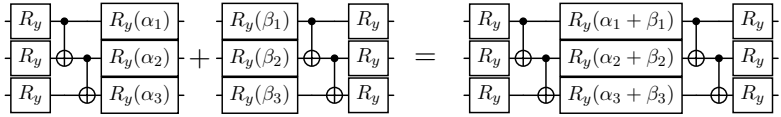
\begin{figure}
    \centering
    	\begin{tikzpicture}
		\node[scale=0.7] at (0,0) {
			\begin{quantikz}[row sep={20pt,between origins}, column sep=3pt]
				 & \gate{R_y}  & \ctrl{1}	&			& \gate{R_y(\alpha_1)} & \\
				 & \gate{R_y}  & \targ{}	& \ctrl{1}	& \gate{R_y(\alpha_2)} &\\
				 & \gate{R_y}  & 			& \targ{}	& \gate{R_y(\alpha_3)} &
			\end{quantikz}
		};
		% +
		\node[] at (1.4,0) {$+$};
		\node[scale=0.7] at (2.8,0) {
			\begin{quantikz}[row sep={20pt,between origins}, column sep=3pt]
				 & \gate{R_y(\beta_1)}  & \ctrl{1}	&			& \gate{R_y}& \\
				 & \gate{R_y(\beta_2)}  & \targ{}	& \ctrl{1}	& \gate{R_y}& \\
				 & \gate{R_y(\beta_3)}  & 			& \targ{}	& \gate{R_y}& 
			\end{quantikz}
		};
		\node[] at (4.4,0) {$=$};
		\node[scale=0.7] at (7,0) {
			\begin{quantikz}[row sep={20pt,between origins}, column sep=3pt]
				& \gate{R_y}  & \ctrl{1}	&			& \gate{R_y(\alpha_1+\beta_1)}  & \ctrl{1}	&			& \gate{R_y} &\\
				& \gate{R_y}  & \targ{}		& \ctrl{1}	& \gate{R_y(\alpha_2+\beta_2)}  & \targ{}	& \ctrl{1}	& \gate{R_y}& \\
				& \gate{R_y}  & 			& \targ{}	& \gate{R_y(\alpha_3+\beta_3)}  & 			& \targ{}	& \gate{R_y} &
			\end{quantikz}
            };
	\end{tikzpicture}
    \caption{Concatenation of Ry-CNOT circuits}
    \label{fig:concatenation-ry-cnot}
\end{figure}

\subsection{Universality}\label{sec:universality}
Section~\ref{sec:decomposition} established that a sequence of HEA circuits can be represented as one HEA circuit with constant depth overhead. The remaining task is to show that any gate in a universal gate set can be represented by a HEA circuit with depth at most polynomial in $N$, where $N$ is the number of qubits of the target circuit.

In this section, we first deal with the Ry-Rz-CZ ansatz.
The key idea is to define a family of circuits $D_k$ and its hermitian conjugate, which turns out to be realized by a poly-depth Ry-Rz-CZ circuit, and show that a combination of $D_k$, $D_k^\dagger$ and $H$ gates are sufficient to realize any nearest-neighbor $CZ$ gates. In combination with the $H$ and $T$ gates that can be realized by a 0-depth Ry-Rz-CZ circuit, each gate in a universal gate set $\{H,T,\CNOT\}$ is shown to be realized efficiently by Ry-Rz-CZ circuits.

Secondly we will prove the universality of Ry-CNOT ansatz. There, we will show that $\CNOTladder^\dagger$ can be efficiently realized by a Ry-CNOT circuit, and combining $Ry, \CNOTladder$ and $\CNOTladder^\dagger$, we will show that nearest-neighbor CNOT gates can be realized by a poly-depth Ry-CNOT circuit, along with an arbitrary Ry gate, which can be realized in 0-depth, completing a computationally universal gate set $\{\Ry,\CNOT\}$.

\subsubsection{$\HEA(R_y\text{-}R_z,\CZladder)$}\label{sec:cz_universality}
In this subsection, we prove each gate in a strictly universal gate set $\{H, T, \mathrm{CZ}\}$ can be realized as a Ry-Rz-CZ circuit of depth polynomial in $N$, where $N$ is the maximum qubit index that the gate acts on.

\paragraph{Single-qubit gates $H$ and $T$}\label{para:cz_one_qubit}\mbox{}\\[0.1em]
The single-qubit gates $H$ and $T$ can be synthesized entirely within the depth-0 the Ry-Rz-CZ circuit, 
which contains a sequence of $R_y$ and $R_z$ rotations at each qubit. The following holds as unitary operations,
\begin{align}
H&= - R_z(\pi)R_y \left(-\frac{\pi}{2}\right), &
T&= e^{i\pi/8} R_z \left(\frac{\pi}{4}\right)R_y(0).
\end{align}

Up to global phases, these expressions show that the $H$ and $T$ gates are equivalent to a sequence of $R_y$ and $R_z$ gates. It is thus possible to realize the action of $H$ or $T$ gates on $k$-th qubit as a 0-depth Ry-Rz-CZ circuit, by setting the angle parameters of $R_y$ and $R_z$ gates at the $k$-th qubit to $-\pi/2$ and $\pi$ for the $H$ gate, and $0$ and $\pi/4$ for the $T$ gate, respectively, and all other parameters to zero.

\paragraph{Two-qubit gate $\CNOT$}\mbox{}\\[0.1em]
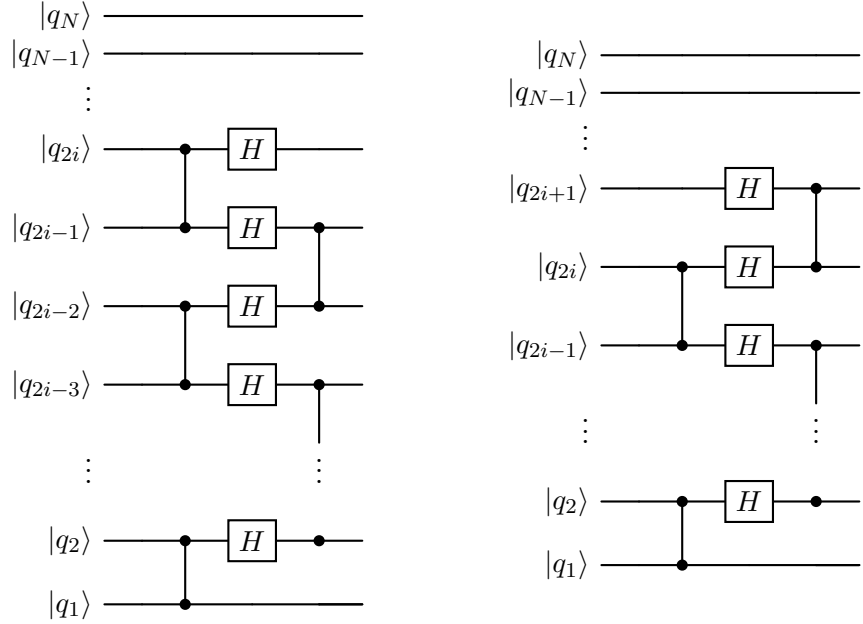
\begin{figure}[tb]
    \makebox[\textwidth]{
\begin{quantikz}
\lstick{$\ket{q_{N}}$}       & &      &      &       & \\
\lstick{$\ket{q_{N-1}}$}     & &      &      &       & \\
\lstick{$\vdots$}   \setwiretype{n}          &     &          &          &           &     \\
\lstick{$\ket{q_{2i}}$}      & & \ctrl{1} & \gate{H} &       & \\
\lstick{$\ket{q_{2i-1}}$}    & & \ctrl{0} & \gate{H} & \ctrl{1}  & \\
\lstick{$\ket{q_{2i-2}}$}    & & \ctrl{1} & \gate{H} & \ctrl{0}  & \\
\lstick{$\ket{q_{2i-3}}$}    & & \ctrl{0} & \gate{H} & \ctrl{1}  & \\
\lstick{$\vdots$} \setwiretype{n}            &     &          &          & \vdots    &     \\
\lstick{$\ket{q_{2}}$}       & & \ctrl{1} & \gate{H} & \ctrl{0}  & \\
\lstick{$\ket{q_{1}}$}       & & \ctrl{0} &      &       & \qw
\end{quantikz}

    \qquad\qquad
\begin{quantikz}
\lstick{$\ket{q_{N}}$}       & &      &      &      & \\
\lstick{$\ket{q_{N-1}}$}     & &      &      &      & \\
\lstick{$\vdots$}   \setwiretype{n}          &     &          &          &          &     \\
\lstick{$\ket{q_{2i+1}}$}    & &      & \gate{H} & \ctrl{1} & \\
\lstick{$\ket{q_{2i}}$}      & & \ctrl{1} & \gate{H} & \ctrl{0} & \\
\lstick{$\ket{q_{2i-1}}$}    & & \ctrl{0} & \gate{H} & \ctrl{1} & \\
\lstick{$\vdots$} \setwiretype{n}            &     &          &          & \vdots   &     \\
\lstick{$\ket{q_{2}}$}       & & \ctrl{1} & \gate{H} & \ctrl{0} & \\
\lstick{$\ket{q_{1}}$}       & & \ctrl{0} &      &      & \qw
\end{quantikz}

    }
    \caption{The left circuit is $D_{2i}$ and the right circuit is $D_{2i+1}$.}
    \label{fig:Di_gate}
\end{figure}
To prove universality, the remaining task is to implement a CNOT gate with the control qubit $k$ and the target qubit $l$, denoted by $\mathrm{CNOT}_{k,l}$. Our construction proceeds in five steps: in Steps 1 to 3, we prepare elements necessary to construct nearest-neighbor CNOT gates. Using those elements, nearest-neighbor CNOT gates are constructed in Step 4, which is used in Step 5 to realize arbitrary CNOT gates.
\newpage
\begin{enumerate}
    \item Single–qubit operations $S$, $S^{\dagger}$ and $SH$
        
        Besides the gates $H$ and $T$ established in the previous section, it can be shown that the operations
        $S$, $S^{\dagger}$ and $SH$ can all be realized by a sequence of Ry and Rz, i.e., as a depth-0 layer of the Ry-Rz-CZ ansatz. Explicit expression is as follows:
        \begin{align}
S
&=e^{i\frac{\pi}{4}}R_z\left(\frac{\pi}{2}\right)R_y(0), &
S^{\dagger}
&=e^{-i\frac{\pi}{4}}R_z\left(-\frac{\pi}{2}\right)R_y(0)\\
SH
&= e^{-i\frac{3}{4}\pi}R_z\left(\frac{3}{2}\pi\right)R_y\left(-\frac{\pi}{2}\right)
\end{align}

    \item Entangling layer $\CZladder$
    
        By setting all variational parameters in a depth-1 HEA layer to zero we obtain the fixed entangling operation $\CZladder$.

    \item $D_{k}$ and $D^{\dagger}_{k}$
    
        Using a polynomial number of the operations that are obtained in Steps 1 and 2, one can synthesize a family of unitaries $D_{k}$ shown in Fig. \ref{fig:Di_gate} together with their Hermitian conjugate $D^{\dagger}_{k}$. 
        Each $D_{k}$ (or $D_{k}^{\dagger}$) costs $5k$ depth-0 and $2k$ depth-1 Ry-Rz-CZ circuits, so they can be represented by an Ry-Rz-CZ circuit of depth $2\times(5k+2k)-2+2k\times 1 + 2=16k$ without any additional qubits.
        Explicit expression is given in Appendix~\ref{subsec:const-Dk}.
\begin{figure}
    	\begin{tikzpicture}
		\node[scale=0.7] at (0,0) {

\begin{quantikz}[row sep={20pt, between origins}, column sep=5pt]
\lstick{$\ket{q_{N}}$}       & &      &      &       & \\
\lstick{$\vdots$}   \setwiretype{n}          &     &          &          &           &     \\
\lstick{$\ket{q_{2i}}$}      & & \ctrl{1} & \gate{H} &       & \\
\lstick{$\ket{q_{2i-1}}$}    & & \ctrl{0} & \gate{H} & \ctrl{1}  & \\
\lstick{$\ket{q_{2i-2}}$}    & & \ctrl{1} & \gate{H} & \ctrl{0}  & \\
\lstick{$\ket{q_{2i-3}}$}    & & \ctrl{0} & \gate{H} & \ctrl{1}  & \\
\lstick{$\vdots$} \setwiretype{n}            &     &          &          & \vdots    &     \\
\lstick{$\ket{q_{2}}$}       & & \ctrl{1} & \gate{H} & \ctrl{0}  & \\
\lstick{$\ket{q_{1}}$}       & & \ctrl{0} &      &       & \qw
\end{quantikz}
		};
		% +
		\node[] at (1.6,0) {$+$};
		\node[scale=0.7] at (3.2,0) {
\begin{quantikz}[row sep={20pt, between origins}, column sep=5pt]
      & &      &      &      & \\
\lstick{$\vdots$}   \setwiretype{n}          &     &          &          &          &     \\
     & &      &      &      & \\
  & &   \ctrl{1}    & \gate{H} & & \\
 & & \ctrl{0} & \gate{H} & \ctrl{1} & \\
  & & \ctrl{1} & \gate{H} & \ctrl{0} & \\
\setwiretype{n}            &     &          &          & \vdots   &     \\
      & & \ctrl{0} & \gate{H} & \ctrl{1} & \\
     & & &      &  \ctrl{0}     & 
\end{quantikz}
		};
        		\node[] at (5,0) {$+$};
		\node[scale=0.7] at (6,0) {
\begin{quantikz}[row sep={20pt, between origins}, column sep=5pt]
      & &  \\
\setwiretype{n}          &    \vdots   & \\
      & \gate{H}&  \\
      & &  \\
      & &  \\
      & &  \\
     \setwiretype{n} &\vdots &  \\
      & &  \\
      & & 
\end{quantikz}
		};
		\node[] at (7,0) {$=$};
		\node[scale=0.7] at (9,0) {
\begin{quantikz}[row sep={20pt, between origins}, column sep=5pt]
      & &  \\
\setwiretype{n}          &    \vdots   & \\
      & \ctrl{1}&  \\
      & \ctrl{0}&  \\
      & &  \\
      & &  \\
     \setwiretype{n} &\vdots &  \\
      & &  \\
      & & 
\end{quantikz}
            };
	\end{tikzpicture}
    \caption{Construction of CZ.}
    \label{fig:construct-cz-1}
\end{figure}
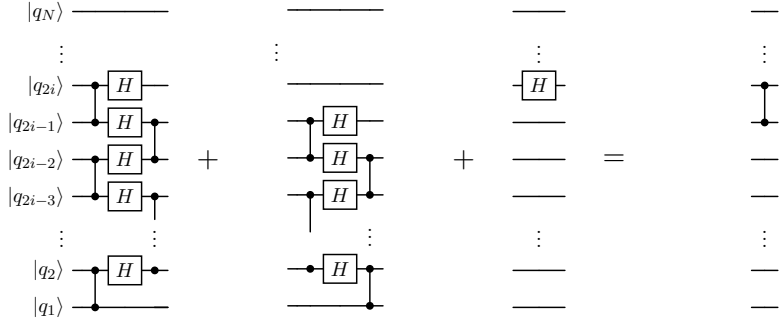

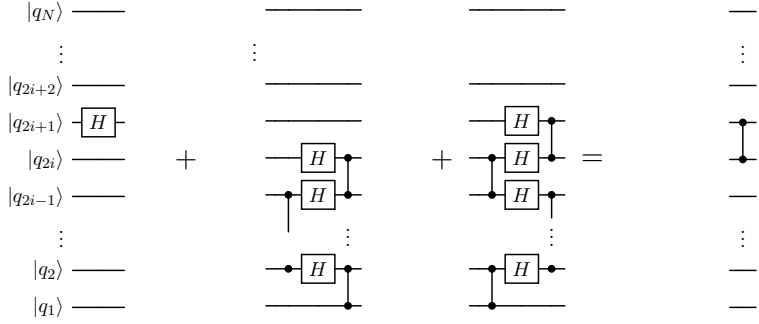
\begin{figure}
    % \centering
    	\begin{tikzpicture}
		\node[scale=0.7] at (0,0) {
\begin{quantikz}[row sep={20pt, between origins}, column sep=5pt]
\lstick{$\ket{q_{N}}$} & & \\
\lstick{$\vdots$}   \setwiretype{n} & & \\
\lstick{$\ket{q_{2i+2}}$}   & & \\
\lstick{$\ket{q_{2i+1}}$}   & \gate{H} & \\
\lstick{$\ket{q_{2i}}$}   & & \\
\lstick{$\ket{q_{2i-1}}$}  & & \\
\lstick{$\vdots$} \setwiretype{n}  & & \\
\lstick{$\ket{q_{2}}$}  & & \\
\lstick{$\ket{q_{1}}$}   & &
\end{quantikz}

		};
		% +
		\node[] at (1.6,0) {$+$};
		\node[scale=0.7] at (3.2,0) {
        
\begin{quantikz}[row sep={20pt, between origins}, column sep=5pt]
      & &      &      &      & \\
\lstick{$\vdots$}   \setwiretype{n}          &     &          &          &          &     \\
     & &      &      &      & \\
  & &      & & & \\
 & & & \gate{H} & \ctrl{1} & \\
  & & \ctrl{1} & \gate{H} & \ctrl{0} & \\
\setwiretype{n}            &     &          &          & \vdots   &     \\
      & & \ctrl{0} & \gate{H} & \ctrl{1} & \\
     & & &      &  \ctrl{0}     & 
\end{quantikz}
		};
        		\node[] at (5,0) {$+$};
		\node[scale=0.7] at (6,0) {
        \begin{quantikz}[row sep={20pt, between origins}, column sep=5pt]
   & &      &      &       & \\
\setwiretype{n}          &     &          &          &           &     \\
 & &  & &       & \\
   & &  & \gate{H} & \ctrl{1}  & \\
   & & \ctrl{1} & \gate{H} & \ctrl{0}  & \\
  & & \ctrl{0} & \gate{H} & \ctrl{1}  & \\
\setwiretype{n}            &     &          &          & \vdots    &     \\
      & & \ctrl{1} & \gate{H} & \ctrl{0}  & \\
    & & \ctrl{0} &      &       & 
\end{quantikz}
		};
		\node[] at (7,0) {$=$};
		\node[scale=0.7] at (9,0) {
\begin{quantikz}[row sep={20pt, between origins}, column sep=5pt]
      & &  \\
\setwiretype{n}          &    \vdots   & \\
      & &  \\
      & \ctrl{1}&  \\
      & \ctrl{0}&  \\
      & &  \\
     \setwiretype{n} &\vdots &  \\
      & &  \\
      & & 
\end{quantikz}
            };
	\end{tikzpicture}
    \caption{Another construction of CZ.}
    \label{fig:construct-cz-2}
\end{figure}
\newcommand{\CZ}{\mathrm{CZ}}
    \item Nearest-neighbor CZ.
        
        Combining $H$, $D_{k}$, and $D_{k}^{\dagger}$, we construct the CZ gate acting on adjacent qubits $k$ and $k+1$, $\mathrm{CZ}_{k,k+1}$.
        As schematically shown in Figs.~\ref{fig:construct-cz-1} and \ref{fig:construct-cz-2}, the following circuit identities hold
        \begin{align}
        \CZ_{2i,2i+1}&=D_{2i+1} D_{2i}^\dagger H_{2i}, &
        \CZ_{2i-1,2i} &= H_{2i} D_{2i-1}^\dagger D_{2i}
        \end{align}
        where $H_k$ denotes the Hadamard gate at the $k$-th qubit. One can check from the explicit expression that $\CZ_{k,k+1}$ can be constructed by a $32k+18$-depth Ry-Rz-CZ circuit. More generally, any nearest-neighbor CZ gate in an $N$-qubit circuit can be realized by a $32N+18$-depth Ry-Rz-CZ circuit.

    \item Arbitrary-distance CNOT.
    
    Finally, let us create arbitary CNOT gates from the operations constructed above. First, an nearest-neighbor CNOT gate can be constructed by adding two $H$ gates on both sides of the nearest-neighbor CZ gate (Fig.~\ref{fig:CZ_CNOT}), realized by a $32k+22$-depth Ry-Rz-CZ circuit. Then, by adding SWAP gates to the nearest-neighbor CNOT gate, one can construct an arbitrary CNOT gate (Fig.~\ref{fig:CNOT_universal}). It is well known that a SWAP gate is realized by three CNOT gates (Fig.~\ref{fig:CNOT_SWAP}). Note that the distance satisfies $\left|k-l\right|\leq \max(k,l)\leq N$, where $N$ is the number of qubits of the target circuit. As can be seen in Fig.~\ref{fig:CNOT_universal}, $2\times(\left|k-l\right|-1)$ SWAPs to propagate the $l$-th qubit to $k+1$-th (or $k-1$-th if $l>k$) qubit and back. In total, the number of nearest-neighbor CNOT gates required for $\mathrm{CNOT}_{k,l}$ is 
        \begin{equation}\label{eq:ryrzcz-6n-5}
            3\times 2\times (\left|k-l\right|-1)+1=6\left|k-l\right|-5\leq 6N-5
        \end{equation}
    We conclude that an arbitrary CNOT gate in an $N$-qubit circuit can be realized by an $N$-qubit Ry-Rz-CZ circuit of depth at most
    \begin{equation}
        2(6N-5)-2+(6N-5)\times (32N+22)=192N^2 - 16N - 122,
    \end{equation}
    which is polynomial in $N$ as claimed.

\end{enumerate}

\begin{figure}[tb]
        \makebox[\textwidth]{
\begin{quantikz}
\lstick{$\ket{q_{N}}$}     & &      & \\
\lstick{$\vdots$} \setwiretype{n}          &     &          &     \\
\lstick{$\ket{q_{i}}$}     & & \ctrl{1} & \\
\lstick{$\ket{q_{i+1}}$}   & & \targ{}  & \\
\lstick{$\vdots$}  \setwiretype{n}         &     &          &     \\
\lstick{$\ket{q_{1}}$}     & &      & \qw
\end{quantikz}

        \quad
        \raisebox{0em}{$=$}
        \quad
        \quad
        \quad
\begin{quantikz}
\lstick{$\ket{q_{N}}$}     & &      &      &      & \\
\lstick{$\vdots$} \setwiretype{n}          &     &          &          &          &     \\
\lstick{$\ket{q_{i+1}}$}   & &      & \ctrl{1} &      & \\
\lstick{$\ket{q_{i}}$}     & & \gate{H} & \ctrl{0} & \gate{H} & \\
\lstick{$\vdots$} \setwiretype{n}          &     &          &          &          &     \\
\lstick{$\ket{q_{1}}$}     & &      &      &      & \qw
\end{quantikz}

        }
    \caption{Construction of a nearest-neighbor CNOT from nearest-neighbor CZ.}
    \label{fig:CZ_CNOT}
    \end{figure}
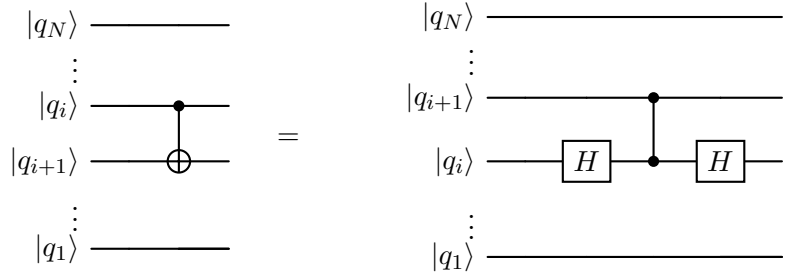

    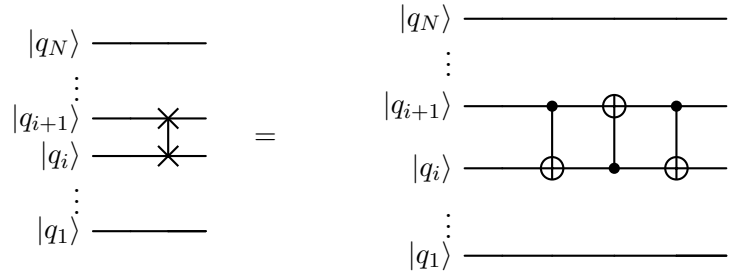
\begin{figure}[tb]
    \centering
        \makebox[\textwidth]{
\begin{quantikz}
\lstick{$\ket{q_{N}}$}     & &        & \\
\lstick{$\vdots$} \setwiretype{n}          &     &            &     \\
\lstick{$\ket{q_{i+1}}$}   & & \swap{1}   & \\
\lstick{$\ket{q_{i}}$}     & & \swap{-1}  & \\
\lstick{$\vdots$} \setwiretype{n}          &     &            &     \\
\lstick{$\ket{q_{1}}$}     & &        & \qw
\end{quantikz}

        \quad
        \raisebox{0em}{$=$}
        \quad
        \quad
        \quad
\begin{quantikz}
\lstick{$\ket{q_{N}}$}     & &      &        &      & \\
\lstick{$\vdots$} \setwiretype{n}          &     &          &            &          &     \\
\lstick{$\ket{q_{i+1}}$}   & & \ctrl{1} & \targ{}    & \ctrl{1} & \\
\lstick{$\ket{q_{i}}$}     & & \targ{}  & \ctrl{-1}  & \targ{}  & \\
\lstick{$\vdots$} \setwiretype{n}          &     &          &            &          &     \\
\lstick{$\ket{q_{1}}$}     & &      &        &      & \qw
\end{quantikz}

        }
    \caption{Construction of a nearest-neighbor SWAP from nearest-neighbor CNOTs.}
    \label{fig:CNOT_SWAP}
    \end{figure}
    
    \begin{figure}[tb]
    \centering
        \makebox[\textwidth]{
\begin{quantikz}
\lstick{$\ket{q_{N}}$}   & & & &      & & & \\
\lstick{$\vdots$}  \setwiretype{n}       &     &     &     &          &     &     &     \\
\lstick{$\ket{q_{i}}$}   & & & & \ctrl{2} & & & \\
\lstick{$\vdots$}  \setwiretype{n}       &     &     &     &          &     &     &     \\
\lstick{$\ket{q_{j}}$}   & & & & \targ{}  & & & \\
\lstick{$\vdots$}  \setwiretype{n}       &     &     &     &          &     &     &     \\
\lstick{$\ket{q_{1}}$}   & & & &      & & & \qw
\end{quantikz}

        \quad
        \raisebox{0em}{$=$}
        \quad
        \quad
        \quad
\begin{quantikz}
\lstick{$\ket{q_{N}}$}     & &      &      &     &      &      & \\
\lstick{$\vdots$}  \setwiretype{n}         &     &          &          &         &          &          &     \\
\lstick{$\ket{q_{i}}$}     & &      &      & \ctrl{1} &     &      & \\
\lstick{$\ket{q_{i-1}}$}   & &      & \swap{1} & \targ{} & \swap{1} &      & \\
\lstick{$\ket{q_{i-2}}$}   & & \swap{1} & \swap{-1}&     & \swap{-1}& \swap{1} & \\
\lstick{$\vdots$}  \setwiretype{n}         &     & \vdots   &          &         &          & \vdots   &     \\
\lstick{$\ket{q_{j}}$}     & & \swap{-1}&      &     &      & \swap{-1}& \\
\lstick{$\vdots$}          &     &          &          &         &          &         &     \\
\lstick{$\ket{q_{1}}$}     & &      &      &     &      &      & \qw
\end{quantikz}
        }
    \caption{Construction of $CNOT_{i,j}$ for arbitrary $i,j$ from nearest-neiaghbor CNOT and SWAPs.}
    \label{fig:CNOT_universal}
\end{figure}
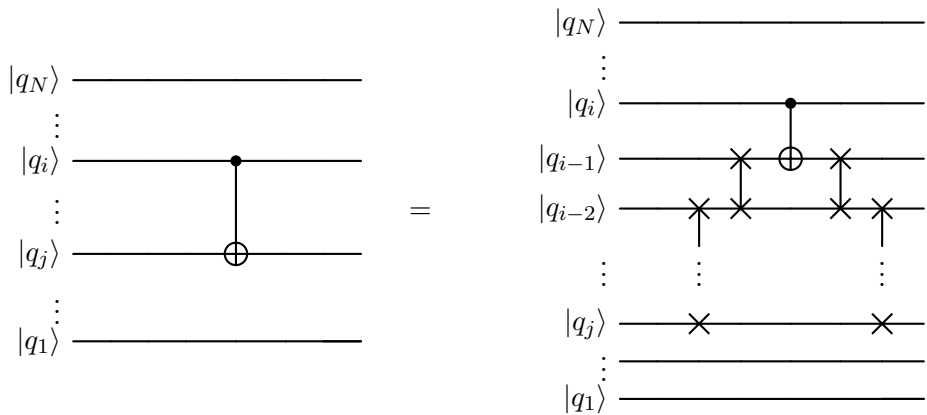

Now that we have explicitly constructed $N$-qubit Ry-Rz-CZ circuits for each gate in a strictly universal gate set $\mathcal{G}=\{H,T,\CNOT\}$, where the depth is $0$ for $H$ and $T$, and $192N^2 - 16N - 122$ for $\CNOT$, one can now construct a $N$-qubit, $\poly(N)$-depth Ry-Rz-CZ circuit representation of any $N$-qubit quantum circuit with $\poly(N)$ gates in $\mathcal{G}$ by merging the ansatz circuits for each gate following Eq.~(\ref{eq:depth-RyRz}). This concludes the proof of Theorem \ref{thm:cz_universal}.

\subsubsection{$\HEA(R_y, \CNOTladder)$}\label{sec:cnot_universality}

In this subsection, we consider a computationally universal gate set $\{R, \mathrm{CNOT}\}$~\cite{R_CNOT_universality}, where $R$ is any \emph{real} single-qubit gate whose square is not diagonal in the computational basis~\cite{R_CNOT_universality}, as we mentioned in Section~\ref{sec:computational_universality_gate_set}. It is actually possible to construct more general gate set $\{R_y, \CNOT\}$ by the Ry-CNOT ansatz, so we focus on this gate set. $R_y$ can be constructed by a depth-0 Ry-CNOT circuit without additional qubits, and CNOT in an $N$-qubit target circuit can be constructed by an $N+4$-qubit Ry-CNOT circuit of depth at most $96N^2+400N-400$.

\paragraph{Single-qubit gate $R_y$}\label{para:cnot_one_qubit}\mbox{}\\[0.1em]
The single-qubit gates $R_y(\theta)$ on $k$-th qubit can be realized by a depth-0 Ry-CNOT circuit by setting the angle parameter of $k$-th qubit to $\theta$ and setting other parameters to zero.

\paragraph{Two-qubit gate CNOT}\label{para:cnot_two_qubit}\mbox{}\\[0.1em]
To complete the proof of the universality, we need a recipe for an arbitrary CNOT gate $\mathrm{CNOT}_{k,l}$.
Our construction proceeds again in five steps: Steps 1-3 prepares necessary ingredients for nearest-neighbor CNOT gates, which is constructed in Step 4 and used in Step 5 to produce arbitrary CNOT gates.
\newpage
\begin{enumerate}
    \item Single–qubit opearations $HX$ and $XH$
        
        The operations $HX$ and $XH$ can be realized by an $R_y$ gate, and thus can be represented as a depth-0 Ry-CNOT circuit.
        \begin{align}
            HX
            &=R_y\left(-\frac{\pi}{2}\right),&
            XH
            &=R_y\left(\frac{\pi}{2}\right)
\end{align}

    \item Entangling layer $\CNOTladder$
    
        By setting all variational parameters in a depth-1 HEA layer to zero we obtain the fixed entangling operation $\CNOTladder$.
        
    \item $N$-qubit inverse CNOT ladder $\CNOTladder^{\dagger}$
        
        As another important ingredient, we show that $\CNOTladder^{\dagger}$ can be realized in the depth-$2^{\,\lceil\log_{2}N\rceil}-1$ Ry-CNOT ansatz. To show that, we first introduce a proposition, which we prove in Appendix~\ref{section:proof-order-of-CNOT-ladder}.

\begin{prop}\label{prop:order-of-CNOT-ladder}
For every integer $N\ge 2$,
$$
   (\CNOTladderN{N})^{\,2^{\lceil\log_{2}N\rceil}} \;=\; E^{(N)},
$$
where $E^{(N)}$ is the $N$-qubit identity.
\end{prop}
This relation reads
\begin{align}\label{ladder_cnot_equality}
   \CNOTladderN{N}^{\dagger}
   \;=\;
   \CNOTladderN{N}^{\dagger}\,
   (\CNOTladderN{N})^{\,2^{\lceil\log_{2}N\rceil}}
   &\;=\;
   (\CNOTladderN{N})^{\,2^{\lceil\log_{2}N\rceil}-1}.
\end{align}
Equation (\ref{ladder_cnot_equality}) demonstrates that $\CNOTladderN{N}^{\dagger}$ can indeed be realized the Ry-CNOT circuit of depth $2^{\lceil\log_{2}N\rceil}-1$ with all the angle parameters set to zero.

    \item Nearest-neighbor $\CNOT$.
    
        Combining $R_y$ (including $HX$ and $XH$), $\CNOTladder$, and $\CNOTladder^{\dagger}$, we construct nearest-neighbor CNOT gates $\mathrm{CNOT}_{k,k+1}$ and $\mathrm{CNOT}_{k+1,k}$.
        More precisely, both $\mathrm{CNOT}_{k,k+1}$ and $\mathrm{CNOT}_{k+1,k}$ can be realized by a $16\times 2^{\lceil\log_{2}(N+4)\rceil}$-depth and $N+4$-qubit Ry-CNOT circuit.
        We leave the rigorous proof of this statement in Appendix~\ref{subsec:adj-CNOT-in-Ry-CNOT}, and outline the proof here. We first show that there is a 6-qubit, Ry-CNOT circuit that realizes $CNOT_{2,1}$
        (Fig~\ref{fig:CNOT_universal_Ry1}). 
        One can see that the circuit is composed of 16 $\CNOTladder$ and 16 $\CNOTladder^\dagger$, and it can be shown that the circuit can be embedded in a circuit with same structure~\footnote{Note that which entangling layer is called $\CNOTladder$ and which is $\CNOTladder^\dagger$ depends on the parity of $k$, but since they are used in pair, the depth count is not affected.} with any number of qubits. 
        This means that, if one adds four~\footnote{6 (the 6-qubit circuit)$-2$ (two-qubit gate)$=4$} ancilla qubits, $\CNOT_{k+1,k}$ for any $k$ can be realized efficiently by an Ry-CNOT circuit. 
        One can use $HX$ and $XH$ to flip the control and target qubits, concluding that one can implement arbitrary nearest-neighbor CNOT gates as $N+4$-qubit Ry-CNOT circuit.
        The total depth is $16\times 1 + 16\times (2^{\lceil\log_{2}(N+4)\rceil}-1)=16\times 2^{\lceil\log_{2}(N+4)\rceil}$.

    \item Arbitrary-distance $\CNOT$.
    
        As we discussed in the previous case (Eq.~(\ref{eq:ryrzcz-6n-5})), $\CNOT_{k,l}$ can be decomposed into $6\left|k-l\right|-5$ nearest-neighbor CNOT gates, and as each nearest-neighbor CNOT gate costs $16\times 2^{\lceil\log_{2}(N+4)\rceil}$ depths in the Ry-CNOT ansatz, $\CNOT_{k,l}$ can be realized by a depth of~\footnote{Recall that concatenating two Ry-CNOT circuits into one does not require any additional depth.} $(6|k-l|-5) \times(16\times 2^{\lceil\log_{2}(N+4)\rceil})< (6N-5)\times (16\times(N+5))=96N^2+400N-400$.

\end{enumerate}

Thus, having demonstrated that every gate in the computationally universal gate set $\mathcal{G}=\{R_y, \CNOT\}$ can be implemented efficiently in the Ry-CNOT ansatz with addition of a constant number (four) of ancilla qubits, one can construct an $N+4$-qubit, $\poly(N)$-depth Ry-CNOT circuit representation of any $N$-qubit target quantum circuit with $\poly(N)$ gates in $\mathcal{G}$ by merging the ansatz circuits for each gate in the target circuit following Eq.~(\ref{eq:depth-Ry}), proving that the Ry-CNOT ansatz is computationally universal (Theorem \ref{thm:cnot_universal}).

\section{Discussion\label{sec:discussion}}
In the previous section, we have shown that the Ry-Rz-CZ ansatz is strictly universal and the Ry-CNOT ansatz is computationally universal, in the sense that we defined at the end of Section~\ref{subsec:universal_gate_sets}. An immediate consequence is that the problem of simulating those ansatz circuits of $N$-qubit and $\poly(N)$-depth is BQP-complete, because (i) we have shown that a BQP-complete problem of simulating a quantum circuit can be reduced to simulating some Ry-Rz-CZ (resp. Ry-CNOT) circuit (BQP-hard) and (ii) simulating a Ry-Rz-CZ (resp. Ry-CNOT) circuit can be efficiently done on quantum computers and thus is in BQP.

This means that simulating those ansatz is, assuming that the polynomial hierarchy does not collapse, classically intractable.
More precisely, it can be shown that the weak classical simulability of one of the two HEA ansatzes within multiplicative error $1\leq c <\sqrt{2}$ would imply that the polynomial hierarchy collapses to the third level, following Corollary~1 of Ref.~\cite{Bremner_2016}. This roughly means in practice that a classical state-vector or sampling simulation of those ansatzes is not practically possible unless a widely believed conjecture in complexity theory is violated.

Note that, as we have shown that Ry-CNOT is hard to simulate, Ry-Rz-CNOT is also hard to simulate, at least in the worst case where all Rz's angles are set to zero. Thus we have shown that Ry-Rz-CZ, Ry-CNOT, and Ry-Rz-CNOT are hard to simulate at least in the worst case. On the other hand, the hardness of Ry-CZ is is an open problem, as we will discuss shortly.

As those particular type of HEA circuits are widely used in (mainly near-term) quantum algorithms such as VQE~\cite{meitei2021gate} and quantum-selected configuration interaction (QSCI)~\cite{kanno2023quantum}, our result gives a theoretical support on those applications that they are very unlikely to be simulable on classical computers.

We also have various directions for future works.
A most straightforward question is if $\HEA(Ry, \CZladder)$ is computationally universal or not. At the moment it is simply that a method to generate the universal gate set as a Ry-CZ circuit is yet to be found, but it would be interesting to see if there is a more fundamental reason that the Ry-CZ ansatz cannot be universal.

It would also be interesting to see if similar results can be shown for broader classes of circuits. HEAs with other types of entangling layers is also widely used, and its universality is of great interest. 
There are various other quantum circuits that are known to be classically sampling hard, such as IQP~\cite{doi:10.1098/rspa.2010.0301}, QAOA~\cite{farhi2019quantumsupremacyquantumapproximate}, match gate circuits~\cite{jozsa2008matchgates, brod2011extending, brod2012geometries,brod2016efficient, hebenstreit2019all}, fermionic systems~\cite{beenakker2004charge,arrazola2022universal}, fermionic linear optics circuits~\cite{divincenzo2005fermionic,oszmaniec2017universal}
as well as chemistry-motivated ansatzes such as the unitary coupled-cluster with singles and doubles (UCCSD)~\cite{1701.02691, 2503.21041} and the unitary cluster Jastrow~\cite{doi:10.1021/acs.jctc.9b00963, 2504.12893} are shown to be sampling-hard. It would be interesting to see if, for example, the symmetry-preserving ansatz~\cite{1904.10910,2002.11724} can also shown to be universal in the sense of Ref.~\cite{2106.13839}.

More practically, it will be interesting to see if we can extend our result to the average-case hardness~\footnote{It is likely that some average-case hardness can be shown for a small subset of the parameter space of HEA, but a more general result would be preferable for practical relevance.}, which is more relevant in practice. In particular, there are some approximate simulation methods in the literature such as Ref.~\cite{2409.01706}, which claims that a certain type of circuits including the HEA circuits, can be simulated (one can compute the expectation value of some observables) efficiently in some accuracy, except for some special cases. There is no conflict between that work and ours, if we think of our worst cases as an exception in Ref.~\cite{2409.01706}, but further study on this would be of great interest.

\section{Conclusion}\label{sec:conclusion}
In this article, we showed that the Ry-Rz-CZ ansatz is strictly universal and the Ry-CNOT ansatz is computationally universal, in the sense that any $N$-qubit quantum circuit with $\poly(N)$ gates in strictly (resp., computationally) universal gate set can be represented as a $\poly(N)$-depth, $\poly(N)$-qubit Ry-Rz-CZ (resp., Ry-CNOT) ansatz. This implies that the problem of simulating those ansatzes is BQP-hard, which is a strong evidence for quantum advantage of simulating those ansatzes.

There are various ways to further explore based on this study, such as studying the classical simulability of the Ry-CZ ansatz, or other HEAs with different single-qubit gates, entangling gates and the structure of the entangling layer.
It would also be practically important to study other popular ansatz such as the symmetry-preserving ansatz~\cite{1904.10910,2002.11724}.
\newpage

\bibliographystyle{apsrev4-2}
\bibliography{apssamp}

\newpage

\appendix

\section{Proof details}

\subsection{Construction of $D_k$ in the Ry-Rz-CZ ansatz}\label{subsec:const-Dk}
We construct the circuit family $D_{k}$, which plays a central role in our synthesis of nearest-neighbor CNOTs,  as a Ry-Rz-CZ circuit. The available primitive operations constructed in the main text are the single-qubit gates
$$
H,S,S^{\dagger},SH,
$$
which are constructed by 0-depth Ry-Rz-CZ circuits, and entangling block
$$
\CZladder,
$$
which is constructed by a 1-depth Ry-Rz-CZ circuit. In the following, we explain a constructive procedure that produces $D_{k}$ for every qubit index $k$ using only a polynomial number of these primitives.
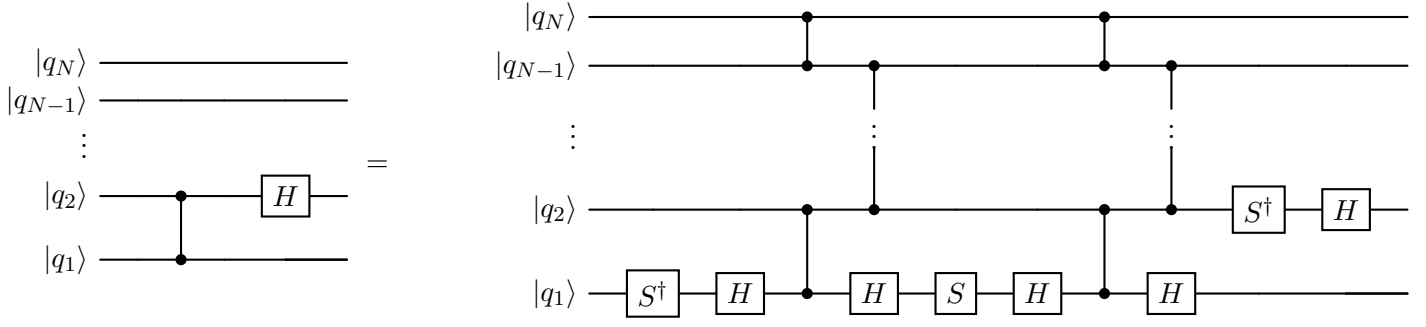
\begin{figure}[tb]
\centering
    \makebox[\textwidth]{
\begin{quantikz}
\lstick{$\ket{q_{N}}$}     & &      & &      & \\
\lstick{$\ket{q_{N-1}}$}   & &      & &      & \\
\lstick{$\vdots$} \setwiretype{n}          &     &          &     &          &     \\
\lstick{$\ket{q_{2}}$}     & & \ctrl{1} & & \gate{H} & \\
\lstick{$\ket{q_{1}}$}     & & \ctrl{0} & &      & \qw
\end{quantikz}

    \raisebox{0.0em}{$=$}
    \quad
    \quad
    \quad
\begin{quantikz}
\lstick{$\ket{q_{N}}$}       &                &      & \ctrl{1} &       &        &        & \ctrl{1} &       &               &     & \\
\lstick{$\ket{q_{N-1}}$}     &                &      & \ctrl{0} & \ctrl{1}  &        &        & \ctrl{0} & \ctrl{1}  &               &     & \\
\lstick{$\vdots$} \setwiretype{n}            &                    &          &          & \vdots    &            &            &          & \vdots    &                   &         &      \\
\lstick{$\ket{q_{2}}$}       &                &      & \ctrl{1} & \ctrl{-1}  &        &        & \ctrl{1} & \ctrl{-1}  & \gate{S^{\dagger}} & \gate{H} & \\
\lstick{$\ket{q_{1}}$}       & \gate{S^{\dagger}} & \gate{H} & \ctrl{0} & \gate{H}  & \gate{S}   & \gate{H}   & \ctrl{0} & \gate{H}  &               &     & \qw
\end{quantikz}

    }
\caption{A quantum circuit identity for $D_1$. On the right, using the fact that $CZ_{i+1,i}$ and $CZ_{i+2,i+1}$ are commutative and CZs acting between the same qubits cancel out each other, all gates except those acting only on $q_1$ and $q_2$ disappear. Then, it can be shown by concrete matrix calculations that the gates acting on $q_1,q_2$ are equivalent to $D_1$ as a whole.}
\label{fig:D1-circuit}
\end{figure}

First we construct $D_1$ gate with $H,SH,HS^{\dagger},\CZladder$. 
\begin{align}
D_{1} &= (H_{2})(S^{\dagger}_{2}H_1)(\CZladder)(H_{1})(S_{1}H_{1})(\CZladder)(H_{1})(S^{\dagger }_{1}),
\end{align}
where subscripts except those for $D_k$ denote the qubit that those single-qubit gates act on.
Figure~\ref{fig:D1-circuit} give an explicit circuit for $D_{1}$. $D_1$ can be implemented with 6 depth-0 circuits and 2 depth-1 circuits, i.e., by an Ry-Rz-CZ circuit of depth $2(6+2)-2+2=16$.

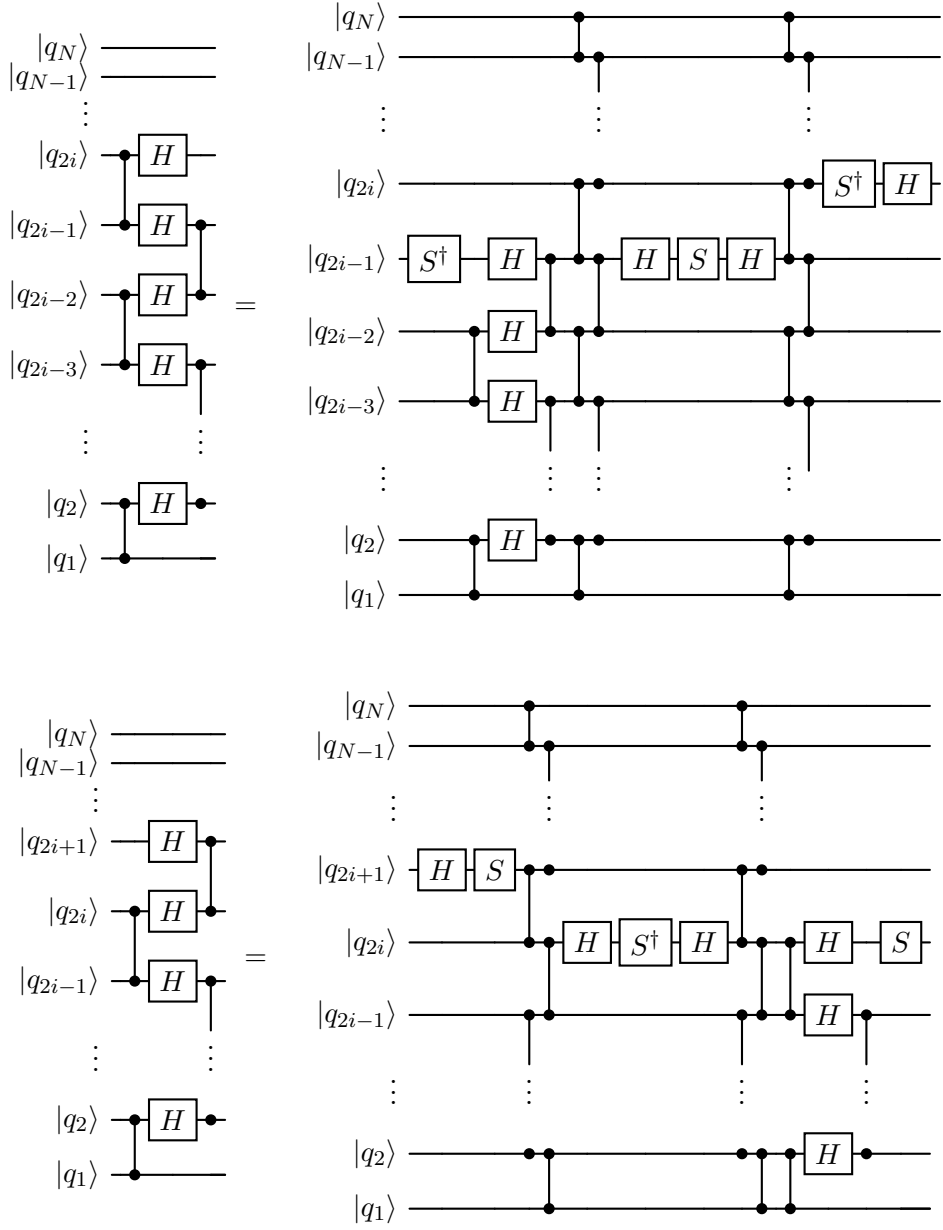
\begin{figure}[tb]
    \makebox[0.8\textwidth]{
\begin{quantikz}[row sep=1.0em, column sep=0.3em]
\lstick{$\ket{q_{N}}$}       & &      &      &       & \\
\lstick{$\ket{q_{N-1}}$}     & &      &      &       & \\
\lstick{$\vdots$} \setwiretype{n}            &     &          &          &           &      \\
\lstick{$\ket{q_{2i}}$}      & & \ctrl{1} & \gate{H} &       & \\
\lstick{$\ket{q_{2i-1}}$}    & & \ctrl{0} & \gate{H} & \ctrl{1}  & \\
\lstick{$\ket{q_{2i-2}}$}    & & \ctrl{1} & \gate{H} & \ctrl{0}  & \\
\lstick{$\ket{q_{2i-3}}$}    & & \ctrl{0} & \gate{H} & \ctrl{1}  & \\
\lstick{$\vdots$} \setwiretype{n}            &     &          &          & \vdots    &      \\
\lstick{$\ket{q_{2}}$}       & & \ctrl{1} & \gate{H} & \ctrl{0}  & \\
\lstick{$\ket{q_{1}}$}       & & \ctrl{0} &      &       & \qw
\end{quantikz}

    \raisebox{0.0em}{$=$}
    \quad
\begin{quantikz}[row sep=1.0em, column sep=0.3em]
\lstick{$\ket{q_{N}}$}       & &      &      &      &      & \ctrl{1}  &       &      &       &       &      & \ctrl{1}  &       &      &       & \\
\lstick{$\ket{q_{N-1}}$}     & &      &      &      &      & \ctrl{0}  & \ctrl{1}  &      &       &       &      & \ctrl{0}  & \ctrl{1}  &      &       & \\
\lstick{$\vdots$} \setwiretype{n}           &     &          &          &          &          &           & \vdots    &          &           &           &          &           & \vdots    &          &           &      \\
\lstick{$\ket{q_{2i}}$}      & &      &      &      &      & \ctrl{1}  & \ctrl{0}  &      &       &       &      & \ctrl{1}  & \ctrl{0}  & \gate{S^{\dagger}} & \gate{H} & \\
\lstick{$\ket{q_{2i-1}}$}    & \gate{S^{\dagger}} & & \gate{H} & \ctrl{1} & & \ctrl{0} & \ctrl{1} & & \gate{H} & \gate{S} & \gate{H} & \ctrl{0} & \ctrl{1} & & & \\
\lstick{$\ket{q_{2i-2}}$}    & & \ctrl{1} & \gate{H} & \ctrl{0} & & \ctrl{1} & \ctrl{0} & & & & & \ctrl{1} & \ctrl{0} & & & \\
\lstick{$\ket{q_{2i-3}}$}    & & \ctrl{0} & \gate{H} & \ctrl{1} & & \ctrl{0} & \ctrl{1} & & & & & \ctrl{0} & \ctrl{1} & & & \\
\lstick{$\vdots$} \setwiretype{n}           &     &          &          & \vdots   &     &         & \vdots    &     &      &      &      & \vdots    &         &      &      &      \\
\lstick{$\ket{q_{2}}$}       & & \ctrl{1} & \gate{H} & \ctrl{0} & & \ctrl{1} & \ctrl{0} & & & & & \ctrl{1} & \ctrl{0} & & & \\
\lstick{$\ket{q_{1}}$}       & & \ctrl{0} &      &      & & \ctrl{0} &      & & & & & \ctrl{0} &      & & & \qw
\end{quantikz}
    }
    \\
    \makebox[\textwidth]{}
    \makebox[\textwidth]{}
    \makebox[\textwidth]{
\begin{quantikz}[row sep=1.0em, column sep=0.3em]
\lstick{$\ket{q_{N}}$}       & &      &      &      & \\
\lstick{$\ket{q_{N-1}}$}     & &      &      &      & \\
\lstick{$\vdots$} \setwiretype{n}           &     &          &          &          &      \\
\lstick{$\ket{q_{2i+1}}$}    & &      & \gate{H} & \ctrl{1} & \\
\lstick{$\ket{q_{2i}}$}      & & \ctrl{1} & \gate{H} & \ctrl{0} & \\
\lstick{$\ket{q_{2i-1}}$}    & & \ctrl{0} & \gate{H} & \ctrl{1} & \\
\lstick{$\vdots$} \setwiretype{n}            &     &          &          & \vdots   &      \\
\lstick{$\ket{q_{2}}$}       & & \ctrl{1} & \gate{H} & \ctrl{0} & \\
\lstick{$\ket{q_{1}}$}       & & \ctrl{0} &      &      & \qw
\end{quantikz}

    \raisebox{0.0em}{$=$}
    \quad
\begin{quantikz}[row sep=1.0em, column sep=0.3em]
\lstick{$\ket{q_{N}}$}       & &      & \ctrl{1} &      &      &            &      & \ctrl{1} &      &      &      &      &      &      & \\
\lstick{$\ket{q_{N-1}}$}     & &      & \ctrl{0} & \ctrl{1} &      &            &      & \ctrl{0} & \ctrl{1} &      &      &      &      &      & \\
\lstick{$\vdots$} \setwiretype{n}            &     &          &          & \vdots   &          &                &          &          & \vdots   &          &          &          &          &          &      \\
\lstick{$\ket{q_{2i+1}}$}    & \gate{H} & \gate{S} & \ctrl{1} & \ctrl{0} & &      &      & \ctrl{1} & \ctrl{0} &      &      &      &      &      & \\
\lstick{$\ket{q_{2i}}$}      & &      & \ctrl{0} & \ctrl{1} & \gate{H} & \gate{S^{\dagger}} & \gate{H} & \ctrl{0} & \ctrl{1} & & \ctrl{1} & \gate{H} & & \gate{S} & \\
\lstick{$\ket{q_{2i-1}}$}    & &      & \ctrl{1} & \ctrl{0} &      &            &      & \ctrl{1} & \ctrl{0} & & \ctrl{0} & \gate{H} & \ctrl{1} & & \\
\lstick{$\vdots$} \setwiretype{n}            &     &          & \vdots   &          &          &                &          & \vdots   &          &     &         &          & \vdots   &     &     \\
\lstick{$\ket{q_{2}}$}       & &      & \ctrl{0} & \ctrl{1} &      &            &      & \ctrl{0} & \ctrl{1} & & \ctrl{1} & \gate{H} & \ctrl{0} & & \\
\lstick{$\ket{q_{1}}$}       & &      &      & \ctrl{0} &      &            &      &      & \ctrl{0} & & \ctrl{0} &      &      & & \qw
\end{quantikz}

    }
    \caption{Circuit identities for $D_{2i}$ (top) and $D_{2i+1}$ (bottom). To go from left to right, using the fact that $CZ_{i+1,i}$ and $CZ_{i+2,i+1}$ are commutative and CZs acting between the same qubits cancel out each other, all gates of $\CZladder$ except those acting only on $q_{2i},q_{2i-1}$ and $q_{2i-2}$ for $D_{2i}$ or $q_{2i+1},q_{2i}$ and $q_{2i-1}$  for $D_{2i+1}$ disappear.}
    \label{fig:Dk+1-circuit}
\end{figure}

Second, assuming that $D_k$ is given, we construct $D_{k+1}$ gate with $D_{k},S,S^{\dagger},H,HS,S^{\dagger}H,\CZladder$. Since we constructed $D_1$ already, this construction enables a construction of $D_k$ for arbitrary $k$. The construction is given as follows:
\begin{align*}
D_{2i} &= (H_{2i})(S^{\dagger}_{2i}) (\CZladder)(H_{2i-1})(S_{2i-1}H_{2i-1})(\CZladder)(D_{2i-1})(S^{\dagger}_{2i-1}) \qquad (k = 2i-1)\\
D_{2i+1} &= (S_{2i+1})(D_{2i})(\CZladder)(H_{2i}S^{\dagger}_{2i})(H_{2i})(\CZladder)(S_{2i+1})(H_{2i+1}) \qquad\qquad\quad (k = 2i)
\end{align*}
Figure \ref{fig:Dk+1-circuit} give an explicit circuit for $D_{k+1}$.  In both cases, a synthesis of $D_{k+1}$ uses four single-qubit gates, two $\CZladder$ and a $D_{k}$. 

Let $\# D_k$ denote the depth of the Ry-Rz-CZ circuit that represents the $D_k$ circuit. By construction, we obtain the relation
\begin{align}
\#D_{k} &= 2(5+2+1)-2+\#D_{k-1}+4\times 0 + 2\times 1=16+\#D_{k-1}
\end{align}
With the initial conditions $\#D_{1} = 16$, one obtains
\begin{align}
\#D_{k} =16(k-1)+\# D_1=16k
\end{align}
Hence $D_{k}$ can always be implemented by an Ry-Rz-CZ circuit of depth $16k$.

Because every primitive gate employed so far is unitary and has its Hermitian conjugate available at identical cost~\footnote{explicitly,
$H^\dagger = H,\ (SH)^\dagger = HS^{\dagger}
,\ (HS^{\dagger})^\dagger = SH,$ and $\CZladder^\dagger =  \CZladder$} the synthesis of $D_k^\dagger$ follows immediately by time-reversing each operation such as $H$ or $\CZladder$ for $D_k$.
$D_k^\dagger$ can thus be implemented with exactly the same cost as $D_k$, i.e., $\#D_k^\dagger =\# D_k$.

\subsection{Nearest-neighbor CNOT gates in the Ry-CNOT ansatz}\label{subsec:adj-CNOT-in-Ry-CNOT}
\begin{figure}[tb]
\begin{quantikz}
\lstick{$\ket{q_6}$}   &\gate{Ry^{(16)}}    &\ctrl{1} &  & &  &\ctrl{1}      &   &  &\ctrl{1}     & &\ctrl{1}  & & \\
\lstick{$\ket{q_5}$}   &\gate{Ry^{(15)}}    &\targ{}  &\ctrl{1}  &\gate{Ry^{(10)}}    &\ctrl{1}  &\targ{}   &\gate{Ry^{(6)}}    &\ctrl{1} &\targ{}   & &\targ{} &\ctrl{1} &  \\
\lstick{$\ket{q_4}$}   &\gate{Ry^{(14)}}    &\ctrl{1} & \targ{}  &\gate{Ry^{(9)}}  &\targ{}   &\ctrl{1}  &  \gate{Ry^{(5)}}   & \targ{}   &\ctrl{1}   &\gate{Ry^{(2)}} &\ctrl{1} & \targ{}   &  \\
\lstick{$\ket{q_3}$}   &\gate{Ry^{(13)}}    &\targ{}  &\ctrl{1}  &\gate{Ry^{(8)}}    &\ctrl{1}  &\targ{}   &\gate{Ry^{(4)}}     &\ctrl{1} &\targ{}   &\gate{Ry^{(1)}}    &\targ{} &\ctrl{1} &  \\
\lstick{$\ket{q_2}$}   &\gate{Ry^{(12)}}    &\ctrl{1} &\targ{}   &\gate{Ry^{(7)}}    &\targ{}   &\ctrl{1}  &     \gate{Ry^{(3)}}      &\targ{}   &\ctrl{1}     & &\ctrl{1} &\targ{}   & \\
\lstick{$\ket{q_1}$}   &\gate{Ry^{(11)}}    &\targ{}  &  & & &\targ{}   &  & &\targ{}   & &\targ{} & &  
\end{quantikz}
    \caption{The quantum circuit structure for each layer $m$ of Eq.~(\ref{eq:cnot21}).}
    \label{fig:CNOT_universal_Ry1}
\end{figure}
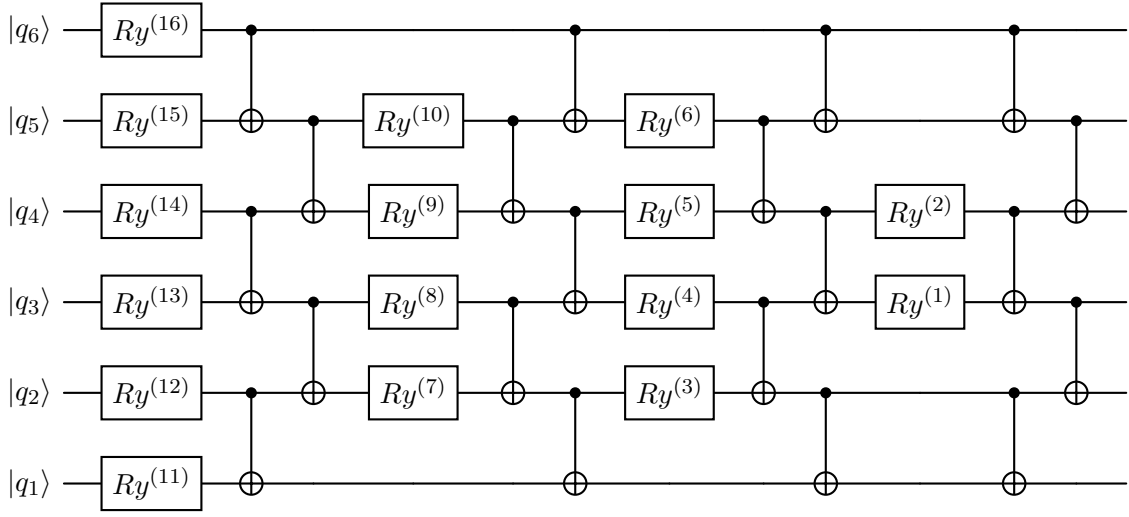

\begin{figure}[tb]
    \makebox[\textwidth]{
\begin{quantikz}
\lstick{$\ket{q_N}$}   \ldots \      &\gate{Ry(0)} &\ctrl{1} &     &\gate{Ry(0)}     &     &\ctrl{1}      &\gate{Ry(0)}      &     &\ctrl{1} &     \gate{Ry(0)}      &\ctrl{1}     &&\gate{Ry(0)} &     \ \cdots \\
\lstick{$\ket{q_{N-1}}$}  \ldots \     &\gate{Ry(0)} &\targ{}  &\ctrl{1}  &\gate{Ry(0)} &\ctrl{1}  &\targ{}   &\gate{Ry(0)}   &\ctrl{1} &\targ{}   &\gate{Ry(0)} &\targ{} &\ctrl{1} &\gate{Ry(0)} & \ \cdots \\
\lstick{$\vdots$}  \setwiretype{n} &&&\vdots &&\vdots&&& \vdots &&&& \vdots  \\
\lstick{$\ket{q_{N-2i}}$}  \ldots \    &\gate{Ry}    &\ctrl{1} &\targ{}   &\gate{Ry(0)} &\targ{}   &\ctrl{1}      &\gate{Ry(0)}    &\targ{}   &\ctrl{1}     &\gate{Ry(0)} &\ctrl{1} &\targ{}   &\gate{Ry}  & \ \cdots \\
\lstick{$\ket{q_{N-2i-1}}$} \ldots \   &\gate{Ry}    &\targ{}  &\ctrl{1}  &\gate{Ry}    &\ctrl{1}  &\targ{}   &\gate{Ry}      &\ctrl{1} &\targ{}   &\gate{Ry(0)} &\targ{} &\ctrl{1} &\gate{Ry}  & \ \cdots \\
\lstick{$\ket{q_{N-2i-2}}$}  \ldots \  &\gate{Ry}    &\ctrl{1} &\targ{}   &\gate{Ry}    &\targ{}   &\ctrl{1}  &     \gate{Ry}     &\targ{}   &\ctrl{1}     &\gate{Ry}    &\ctrl{1} &\targ{}   &\gate{Ry} & \ \cdots \\
\lstick{$\ket{q_{N-2i-3}}$} \ldots \   &\gate{Ry}    &\targ{}  &\ctrl{1}  &\gate{Ry}    &\ctrl{1}  &\targ{}   &\gate{Ry}     &\ctrl{1} &\targ{}   &\gate{Ry}    &\targ{} &\ctrl{1} &\gate{Ry} & \ \cdots \\
\lstick{$\ket{q_{N-2i-4}}$}  \ldots \  &\gate{Ry}    &\ctrl{1} &\targ{}   &\gate{Ry}    &\targ{}   &\ctrl{1}  &     \gate{Ry}      &\targ{}   &\ctrl{1}     &\gate{Ry(0)} &\ctrl{1} &\targ{}   &\gate{Ry} & \ \cdots \\
\lstick{$\ket{q_{N-2i-5}}$} \ldots \   &\gate{Ry}    &\targ{}  &\ctrl{1}  &\gate{Ry(0)} &\ctrl{1}  &\targ{}   &\gate{Ry(0)}     &\ctrl{1} &\targ{}   &\gate{Ry(0)} &\targ{} &\ctrl{1} &\gate{Ry} & \ \cdots \\
\setwiretype{n} \vdots &&&\vdots & & \vdots &&& \vdots &&&&\vdots
\end{quantikz}
}
=
\begin{quantikz}
\lstick{$\ket{q_N}$}   \ldots \      & & &     &     &     &     &      &     & &           &    && &     \ \cdots \\
\lstick{$\ket{q_{N-1}}$}  \ldots \     & &&  & &&&   & &   & & & & & \ \cdots \\
\lstick{$\vdots$}  \setwiretype{n} &&& &&&&&  &&&&   \\
\lstick{$\ket{q_{N-2i}}$}  \ldots \    &\gate{Ry}    &\ctrl{1} &  & &  &\ctrl{1}      &   &  &\ctrl{1}     & &\ctrl{1} & &\gate{Ry}  & \ \cdots \\
\lstick{$\ket{q_{N-2i-1}}$} \ldots \   &\gate{Ry}    &\targ{}  &\ctrl{1}  &\gate{Ry}    &\ctrl{1}  &\targ{}   &\gate{Ry}      &\ctrl{1} &\targ{}   & &\targ{} &\ctrl{1} &\gate{Ry}  & \ \cdots \\
\lstick{$\ket{q_{N-2i-2}}$}  \ldots \  &\gate{Ry}    &\ctrl{1} &  \targ{}  &\gate{Ry}    &\targ{}    &\ctrl{1}  &     \gate{Ry}     & \targ{}   &\ctrl{1}     &\gate{Ry}    &\ctrl{1} & \targ{}   &\gate{Ry} & \ \cdots \\
\lstick{$\ket{q_{N-2i-3}}$} \ldots \   &\gate{Ry}    &\targ{}  &\ctrl{1}  &\gate{Ry}    &\ctrl{1}  &\targ{}   &\gate{Ry}     &\ctrl{1} &\targ{}   &\gate{Ry}    &\targ{} &\ctrl{1} &\gate{Ry} & \ \cdots \\
\lstick{$\ket{q_{N-2i-4}}$}  \ldots \  &\gate{Ry}    &\ctrl{1} &\targ{}   &\gate{Ry}    &\targ{}   &\ctrl{1}  &     \gate{Ry}      &\targ{}   &\ctrl{1}     & &\ctrl{1} &\targ{}   &\gate{Ry} & \ \cdots \\
\lstick{$\ket{q_{N-2i-5}}$} \ldots \   &\gate{Ry}    &\targ{}  &  & & &\targ{}   &    & &\targ{}   & &\targ{} & &\gate{Ry} & \ \cdots \\
\setwiretype{n} \vdots &&& & & &&& \vdots &&&&
\end{quantikz}
    \caption{By setting some of the parameters of Ry gates to zero as shown in the circuit above, the circuit of 6-qubit circuit like Eq.~(\ref{eq:cnot21}) and Fig.~\ref{fig:CNOT_universal_Ry1} can be constructed between ($N-2i$, $N-2i-5$)-qubits.}
    \label{fig:nn_cnot_create}
\end{figure}
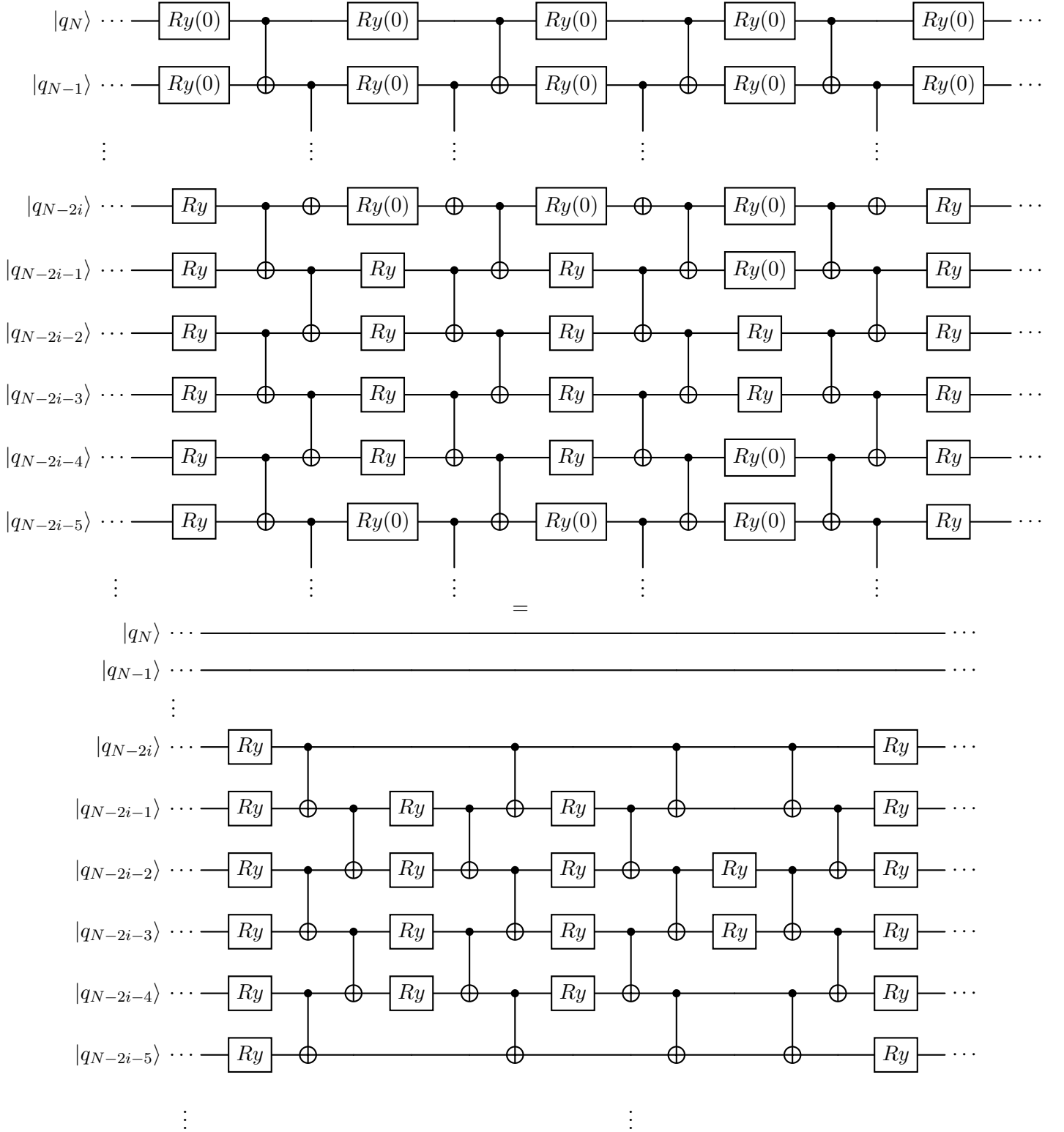
For a 6-qubit Ry-CNOT circuit, a nearest-neighbor CNOT gate $\mathrm{CNOT}_{2,1}$ can be implemented by choosing the rotation angles of the single-qubit gates $R_y^{(k)}$ in the following way:
\begin{align}
\label{eq:cnot21}
\mathrm{CNOT}_{2,1} = (\prod_{i=1}^6 R_y^{(i)}(\theta_{0,i}))\prod_{m=1}^{8} & \Bigl( \CNOTladder \cdot (\prod_{i=1}^{2} R_y^{(i+2)}(\theta_{m,i}))\cdot\CNOTladder^{\dagger} \cdot (\prod_{i=1}^4 R_y^{(i+1)}(\theta_{m,i+2})) \nonumber\\
&\cdot  \CNOTladder^{\dagger} \cdot (\prod_{i=1}^4 R_y^{(i+1)}(\theta_{m,i+6}))\cdot \CNOTladder \cdot (\prod_{i=1}^6 R_y^{(i)}(\theta_{m,i+10})) \Bigr),
\end{align}
where $R_y^{(i)}$ denotes the $R_y$ gate acting on the $i$-th qubit.
The angle values are given in Table~\ref{tab:cnot21-ry-cnot-angles} and the circuit structure for each $m$ is shown in Fig.~\ref{fig:CNOT_universal_Ry1}.
\newcommand{\myfrac}[2]{#1/#2}
\begin{table}[h]
\caption{Explicit values of $\theta_{m,i}$ in Eq.~(\ref{eq:cnot21}) in the unit of $\pi$}
\setlength{\tabcolsep}{12pt}
\centering
\begin{tabular}{cccccccccc}
$i \setminus m$ & 0 & 1 & 2 & 3 & 4 & 5 & 6 & 7 & 8 \\
\hline
1 & 1 & $\myfrac{1}{2}$ & 0 & 0 & $\myfrac{3}{4}$ & 1 & 0 & $\myfrac{1}{2}$ & $\myfrac{1}{2}$ \\
2 & $\myfrac{1}{2}$ & 0 & $\myfrac{1}{2}$ & $\myfrac{1}{4}$ & $\myfrac{3}{4}$ & $\myfrac{1}{2}$ & $\myfrac{1}{2}$ & 0 & $\myfrac{1}{4}$ \\
3 & 0 & 0 & $\myfrac{1}{4}$ & $\myfrac{1}{2}$ & $\myfrac{1}{4}$ & $\myfrac{1}{2}$ & 0 & $\myfrac{1}{2}$ & 0 \\
4 & 1 & $\myfrac{1}{2}$ & $\myfrac{1}{2}$ & $\myfrac{1}{2}$ & $\myfrac{1}{2}$ & $\myfrac{1}{2}$ & 0 & 0 & $\myfrac{1}{2}$ \\
5 & $\myfrac{1}{4}$ & $\myfrac{1}{2}$ & $\myfrac{1}{2}$ & 0 & $\myfrac{1}{2}$ & 0 & $\myfrac{1}{4}$ & $\myfrac{1}{2}$ & 0 \\
6 & $\myfrac{3}{4}$ & $\myfrac{1}{2}$ & $\myfrac{1}{2}$ & $\myfrac{1}{2}$ & $\myfrac{1}{2}$ & 1 & $\myfrac{1}{2}$ & $\myfrac{1}{2}$ & $\myfrac{1}{4}$ \\
7 &  & $\myfrac{1}{2}$ & $\myfrac{1}{2}$ & $\myfrac{1}{2}$ & $\myfrac{1}{2}$ & 0 & $\myfrac{1}{2}$ & $\myfrac{1}{2}$ & 1 \\
8 &  & $\myfrac{1}{2}$ & 0 & 0 & 0 & 0 & $\myfrac{1}{2}$ & 0 & $\myfrac{1}{4}$ \\
9 &  & 0 & $\myfrac{1}{2}$ & $\myfrac{1}{2}$ & $\myfrac{1}{2}$ & $\myfrac{1}{2}$ & 0 & $\myfrac{1}{2}$ & 0 \\
10 &  & 0 & $\myfrac{7}{2}$ & 0 & 0 & 0 & 0 & 0 & $\myfrac{1}{2}$ \\
11 &  & $\myfrac{1}{2}$ & $\myfrac{3}{4}$ & $\myfrac{1}{2}$ & 1 & 1 & 1 & 1 & 0 \\
12 &  & $\myfrac{1}{2}$ & $\myfrac{3}{4}$ & $\myfrac{1}{2}$ & $\myfrac{1}{4}$ & $\myfrac{1}{2}$ & $\myfrac{1}{2}$ & 0 & $\myfrac{1}{2}$ \\
13 &  & $\myfrac{1}{2}$ & $\myfrac{1}{2}$ & $\myfrac{1}{2}$ & $\myfrac{1}{2}$ & 0 & 0 & $\myfrac{1}{2}$ & 0 \\
14 &  & 0 & 0 & 1 & $\myfrac{1}{4}$ & 0 & $\myfrac{1}{2}$ & $\myfrac{1}{2}$ & 0 \\
15 &  & $\myfrac{1}{2}$ & $\myfrac{1}{2}$ & $\myfrac{1}{2}$ & $\myfrac{1}{2}$ & 0 & 0 & 0 & $\myfrac{1}{2}$ \\
16 &  & 1 & 1 & 1 & 0 & 0 & 1 & 1 & $\myfrac{1}{4}$ \\
\end{tabular}
\label{tab:cnot21-ry-cnot-angles}
\end{table}

Noticing that $\CNOTladder$ and $\CNOTladder^\dagger$ are always used in pairs, it is possible to embed the above 6-qubit circuits into an Ry-CNOT circuit of any qubits $N$ with $N\ge 6$ simply by choosing suitable angles, because successive CNOTs acting on the same qubits cancel each other (Fig. \ref{fig:nn_cnot_create}).

More precisely, the gate $\mathrm{CNOT}_{k,k-1}$ is obtained by embedding Eq.~(\ref{eq:cnot21}) to $q_{k-1}, q_k, \dots, q_{k+4}$. Note that, for $k>N-4$, where $N$ is the number of qubits of the target circuit, at most four additional qubits is required to accommodate the 6-qubit circuit. For example, to implement $\CNOT_{N,N-1}$, the ansatz circuit should be embedded in $(q_{N-1},\dots, q_{N+4})$. 

Let us count the depth to realize the nearest-neighbor CNOT gates. 
Eq.~(\ref{eq:cnot21}) uses 16 $\CNOTladder$ (depth-1) and 16 $\CNOTladder^{\dagger}$ (depth-$2^{\lceil\log_{2}N\rceil}-1$ for $N$-qubit; see Appendix~\ref{section:proof-order-of-CNOT-ladder}).

Moreover, because the compound gates $HX$ and $XH$ can be realised within a zero-depth layer, they enable us to reverse the control–target direction of a CNOT with no extra cost. We conclude that any nearest-neighbor CNOT gate for an $N$-qubit target circuit can be realized by a $16\times 2^{\lceil\log_{2}(N+4)\rceil}$-depth, $N+4$-qubit Ry-CNOT circuit.

\subsection{Proof of Proposition~\ref{prop:order-of-CNOT-ladder}}\label{section:proof-order-of-CNOT-ladder}
\begin{prop*}[\ref{prop:order-of-CNOT-ladder}]
For every integer $N\ge 2$,
$$
   (\CNOTladder)^{\,2^{\lceil\log_{2}N\rceil}} \;=\; E^{(N)},
$$
where $E^{(N)}$ is the $N$-qubit identity.
\end{prop*}

\begin{proof}
We prove the claim by mathematical induction on the number of qubits $N$.
For the base case $N=2$ we have
\begin{align}
(\CNOTladderN{2})^{2}=(\CNOT)^{2}=E^{(2)},
\end{align}
so the statement holds.
Assume as the induction hypothesis that for some $N\ge 2$
\begin{align}
\bigl(\CNOTladderN{N}\bigr)^{2^{\lceil\log_{2}N\rceil}}=E^{(N)}.
\end{align}
To extend the result to $N+1$ qubits, add the CNOT gates acting between on $(N\!+\!1)$-th (control) and $N$-th (target) qubits between each adjacent pair of $N$-qubit ladders. Note that we need to invert the $\CNOTladderN{N}$ because the definition of the ladder depends on the parity of $N$.
\begin{align}
\CNOTladderN{N+1} &= \CNOTladderN{N}^{\dagger}\mathrm{CNOT}_{N+1,N} \nonumber\\
&= \CNOTladderN{N}^{\dagger}
(\ket{0}\bra{0}_{N+1} \otimes E^{(N)} + \ket{1}\bra{1}_{N+1} \otimes (E^{(N-1)} \otimes X_N) \nonumber\\
 &= \ket{0}\bra{0}_{N+1} \otimes \CNOTladderN{N}^{\dagger} + \ket{1}\bra{1}_{N+1} \otimes \CNOTladderN{N}^{\dagger}X_N
\end{align}
Raising this operator to the power $2^{\lceil\log_{2}(N+1)\rceil}$ yields
\begin{align}\label{eq:B4}
&(\CNOTladderN{N+1})^{2^{\lceil\log_{2}(N+1)\rceil}} \nonumber\\
&= (\ket{0}\bra{0}_{N} \otimes \CNOTladderN{N}^{\dagger} + \ket{1}\bra{1}_{N} \otimes \CNOTladderN{N}^{\dagger}X_N)^{2^{\lceil\log_{2}(N+1)\rceil}} \nonumber\\
&= \ket{0}\bra{0}_{N} \otimes (\CNOTladderN{N}^{\dagger})^{2^{\lceil\log_{2}(N+1)\rceil}} + \ket{1}\bra{1}_{N} \otimes (\CNOTladderN{N}^{\dagger}X_N)^{2^{\lceil\log_{2}(N+1)\rceil}}
\end{align}

From the fact that for some integer $l = 0,1$ one can write

\begin{align}
\lceil\log_{2}(N+1)\rceil \;=\; \lceil\log_{2}(N)\rceil + l,
\end{align}

Using the induction hypothesis $\bigl(\CNOTladderN{N}\bigr)^{2^{\lceil\log_{2}N\rceil}}=E_{N}$, we obtain the following equation,
\begin{align}
\bigl(\CNOTladderN{N}^{\dagger}\bigr)^{2^{\lceil\log_{2}(N+1)\rceil}}=\bigl(\CNOTladderN{N}^{\dagger}\bigr)^{2^{(\lceil\log_{2}(N)\rceil+l)}}=\bigl(\bigl(\CNOTladderN{N}^{\dagger}\bigr)^{2^{(\lceil\log_{2}(N)\rceil)}}\bigr)^{2^l}=E^{(N)},
\end{align}
Therefore, to show Eq.~(\ref{eq:B4}), it is sufficient to prove the following statement
\begin{align}\label{eq:X_Ladder_eq}
\bigl(X_{N}\,\CNOTladderN{N}\bigr)^{2^{\lceil\log_{2}(N+1)\rceil}}=E^{(N)},
\end{align}
for once this identity holds, both terms in Eq. (\ref{eq:B4})reduce to the identity and we obtain
\begin{align}
\bigl(\CNOTladderN{N+1}\bigr)^{2^{\lceil\log_{2}(N+1)\rceil}}=E^{(N+1)}.
\end{align}

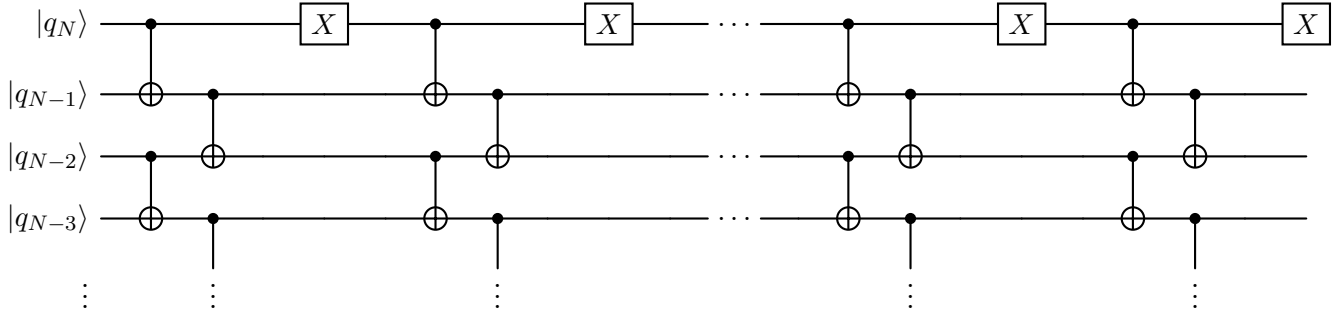
\begin{figure}[tb]
\begin{quantikz}
    \lstick{$\ket{q_{N}}$}    
    &  \ctrl{1} & & 
    & \gate{X}
    & & \ctrl{1} & & 
    & \gate{X} 
    & 
    & \ \dots \ 
    & & \ctrl{1} & & 
    & \gate{X}
    & & \ctrl{1} & & 
    & \gate{X}
    \\
    \lstick{$\ket{q_{N-1}}$} 
    & \targ{}& \ctrl{1} & 
    & 
    & & \targ{}& \ctrl{1} & 
    & 
    & 
    & \ \dots \ 
    & & \targ{}& \ctrl{1} & 
    & 
    & & \targ{}& \ctrl{1} & 
    & 
    \\
    \lstick{$\ket{q_{N-2}}$} 
    & \ctrl{1} & \targ{}& 
    & 
    & & \ctrl{1} & \targ{}& 
    & 
    & 
    & \ \dots \ 
    & & \ctrl{1} & \targ{}& 
    & 
    & & \ctrl{1} & \targ{}& 
    & 
    \\
    \lstick{$\ket{q_{N-3}}$} 
    & \targ{}& \ctrl{1} & 
    & 
    & & \targ{}& \ctrl{1} & 
    &     
    & 
    & \ \dots \ 
    & & \targ{}& \ctrl{1} & 
    & 
    & & \targ{}& \ctrl{1} & 
    & 
    \\
    \lstick{\vdots} \setwiretype{n}
    & & \vdots & 
    &
    & & & \vdots & 
    & 
    & &
    & & & \vdots & 
    &
    & & & \vdots & 
\end{quantikz}
    \caption{$N$-qubit, $2^{k+1}$-depth circuit.}
    \label{fig:Cdag_X}
\end{figure}

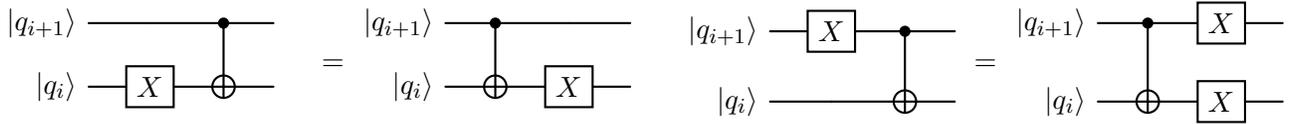
\begin{figure}[tb]
    \centering
    \makebox[\textwidth]{
\begin{quantikz}
    \lstick{$\ket{q_{i+1}}$} & & \ctrl{1}   &  \\
    \lstick{$\ket{q_{i}}$} & \gate{X} & \targ & &
    \end{quantikz}
    \quad
    $=$
    \begin{quantikz}
    \lstick{$\ket{q_{i+1}}$} & \ctrl{1} &   &  \\
    \lstick{$\ket{q_{i}}$} & \targ{} & \gate{X} & 
    \end{quantikz}
    \quad
    \begin{quantikz}
    \lstick{$\ket{q_{i+1}}$} & \gate{X} & \ctrl{1}   &  \\
    \lstick{$\ket{q_{i}}$} & & \targ{} &
    \end{quantikz}
    $=$
    \begin{quantikz}
    \lstick{$\ket{q_{i+1}}$} & \ctrl{1} &  \gate{X}  &  \\
    \lstick{$\ket{q_{i}}$} & \targ{} & \gate{X} & 
\end{quantikz}
}
\caption{$X$ propagation}
\label{fig:X_propagation}
\end{figure}

We analyse the operator $\bigl(X_{N}\CNOTladderN{N}\bigr)^{k}$ in Fig.~\ref{fig:Cdag_X}.
Starting from this figure, we propagate each $X_N$ gate to the right of all $\CNOTladderN{N}$ gates according to the rules shown in Fig.~\ref{fig:X_propagation}.
At the end, we will see that the $X$ count at each qubit is always even and cancels out for the specific $k$ of our interest, concluding the proof of Eq.~(\ref{eq:X_Ladder_eq}).

\begin{figure}[tb]
    \includegraphics[width=\textwidth]{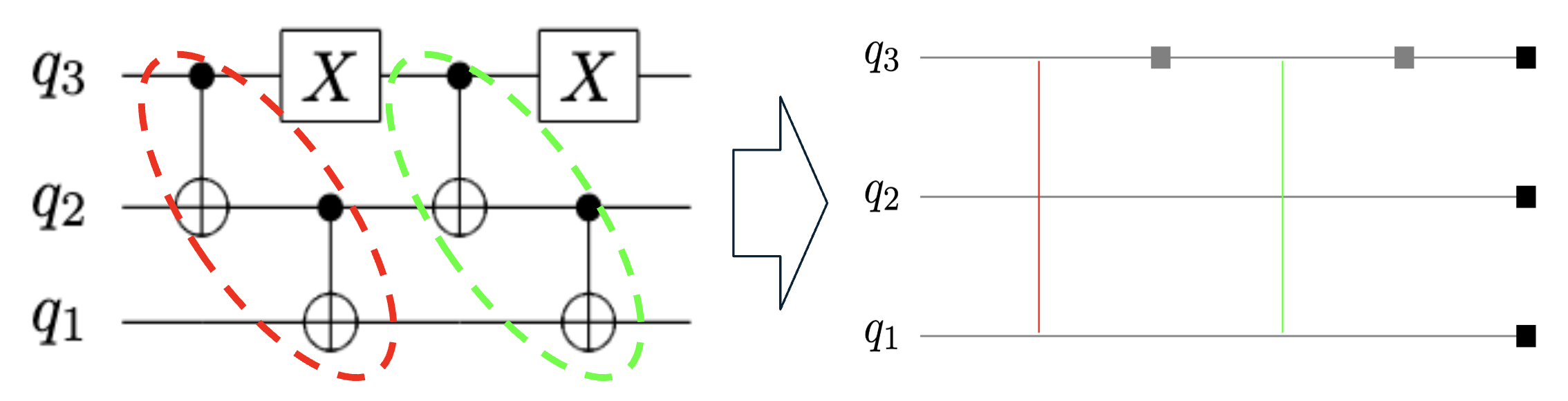}
    \caption{An example of converting the propagation problem into the shortest path problem.}
    \label{fig:X_propagation_2_2}
\end{figure}
\begin{figure}[tb]
    \includegraphics[width=\textwidth]{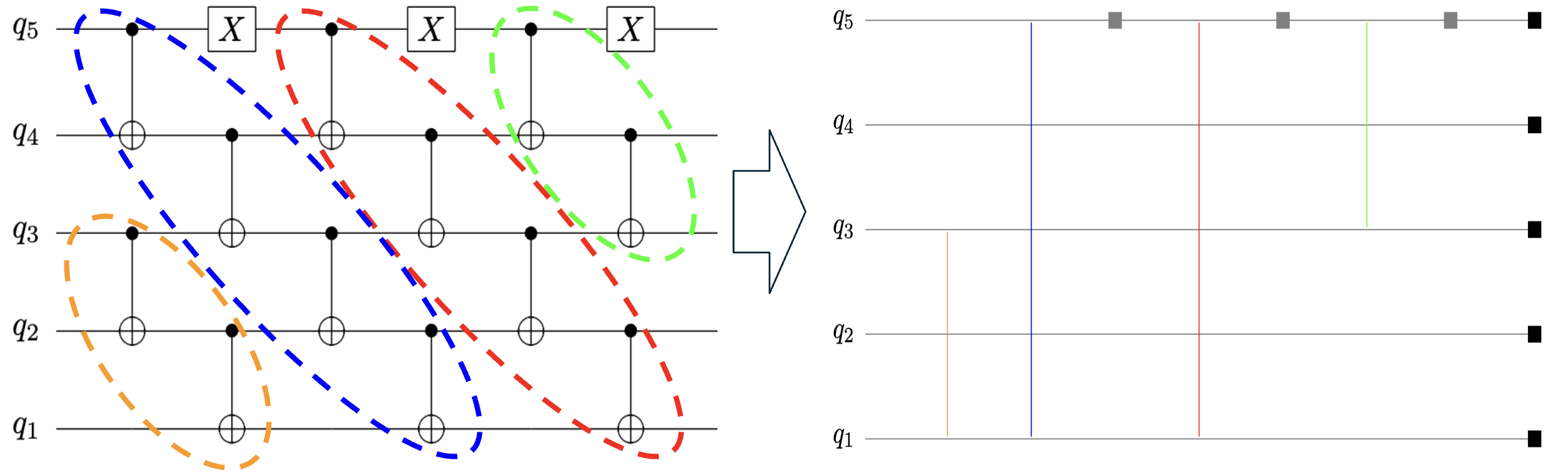}
    \caption{Another example of converting the propagation problem into the shortest path problem.}
    \label{fig:X_propagation_5_3}
\end{figure}

Let us first explain the propagation works in the concrete instance $N=3,\,k=2$ shown in Fig.~\ref{fig:X_propagation_2_2}.
The left panel gives the circuit diagram, the right panel its lattice diagram. In the lattice diagram on the right, the gray squares denote the $X$ gates in the circuit diagram on the left, and the black squares mark the positions those $X$ gates occupy after propagation.  The red and green vertical lines highlight the CNOTs enclosed by red and green boxes, respectively, in the left-hand circuit diagram.  In this picture, the total number of propagated $X$ gates at each black square turns out to be equal to the number of shortest paths (moving only rightward or downward) from any gray square to the black square. From the rightmost gray square to the black square at $q_3$ there is exactly one such path, and similarly one path from the left gray square. Only the left gray square can reach the black squares at $q_1$ and $q_2$, and in each case there is exactly one path passing through the green line. CNOTs represented by the red vertical line lie strictly to the left of every admissible path and therefore do not influence the count. Larger circuits, e.g. Fig~\ref{fig:X_propagation_5_3}, behave in the same way. In those cases, the X count on the right is not necessarily even.

For the general case we label the initial $X$ gates from the right, setting $i=0$ for the rightmost one.  Let $q_j\quad (j=1,\dots ,N)$ denotes the destination qubit.
A path from $i$-th $X$ to $q_j$ exists if and only if the $i$-th $X$ is not farther right than the first CNOT layer that can propagate it to $q_j$, i.e.
\begin{align}
i\;\ge\;\Bigl\lfloor\tfrac{j-1}{2}\Bigr\rfloor.
\end{align}

For $i$ and $j$ that satisfies the above relation, the lattice that we count the paths has height $j-1$ (the downward steps necessary to reach $q_j$) and width
$i-\bigl\lfloor\frac{j-1}{2}\bigr\rfloor$ (number of ``ladders'' available to reach $q_j$).
The number of paths from $i$-th $X$ gates on qubit $j$ is 
\begin{align}
\binom{i-\bigl\lfloor\frac{j-1}{2}\bigr\rfloor+j-1}{j-1} = 
\binom{i+\bigl\lfloor\frac{j}{2}\bigr\rfloor}{j-1},
\end{align}
and it is zero when the inequality fails. Using the well-known \emph{hockey-stick} identity,  summing of k $X$ gates from each layer gives
\begin{align}
\operatorname{count}_{X}(k,j)=\sum_{i=0}^{k-1}\binom{i+\bigl\lfloor\frac{j}{2}\bigr\rfloor}{j-1}
=\binom{k+\bigl\lfloor\frac{j}{2}\bigr\rfloor}{\,j},
\end{align}
which gives the exact number of $X$ gates landing on qubit $q_j$ for arbitrary $N$ and $k$.

Choose the exponent $
k = 2^{\lceil\log_{2}(N+1)\rceil}$,
\begin{align}
\operatorname{count}_{X}(2^{\lceil\log_{2}(N+1)\rceil},j)=
\binom{%
      2^{\lceil\log_{2}(N+1)\rceil}+\bigl\lfloor\frac{j}{2}\bigr\rfloor}
      {\,j}.
\end{align}
To prove Eq.~(\ref{eq:X_Ladder_eq}), it suffices to show that
$\operatorname{count}_{X}\bigl(2^{\lceil\log_{2}(N+1)\rceil},j\bigr)$
is even for all $j$.  For this purpose, we invoke Lucas' theorem.
In binary form it claims the following:
For
\begin{align}
A&=\sum a_{s}2^{s}, &
B&=\sum b_{s}2^{s} & (a_{s},b_{s}\in\{0,1\}),
\end{align}
the following equivalence holds:
\begin{align}
\binom{A}{B}\equiv1\pmod{2}
\Longleftrightarrow
b_{s}\le a_{s}\text{ for every }s.
\end{align}

Examine the two arguments $A, B$ of the binomial coefficient in our case.
\begin{align}
A=2^{\lceil\log_{2}(N+1)\rceil}+\Bigl\lfloor\frac{j}{2}\Bigr\rfloor.
\end{align}
The binary representation for \(2^{\lceil\log_{2}(N+1)\rceil}\) is a single $1$ in position $\lceil\log_{2}(N+1)\rceil$ followed by zeros.

Since
\begin{align}
B=j \;<\; 2^{\lceil\log_{2}(N+1)\rceil}
\quad\text{and}\quad
\lfloor\tfrac{j}{2}\rfloor  \;\leq\;\tfrac{j}{2}=B/2,
\end{align}
it follows that the largest index $s$ for which $b_s = 1$ must have $a_s = 0$, because $s$ is smaller than $\lceil\log_{2}(N+1)\rceil$, which is the only bit that \(2^{\lceil\log_{2}(N+1)\rceil}\) contributes, while $s$ is shown, in the latter part of the inequalities, to be larger than the top bit that $\lfloor\tfrac{j}{2}\rfloor$ in $A$ contributes.
We conclude that
\begin{align}
\operatorname{count}_{X}\bigl(2^{\lceil\log_{2}(N+1)\rceil},j\bigr)
      \equiv 0\pmod{2}\quad
      (j=1,\dots ,N).
\end{align}

Thus after all propagated \(X\) gates of the circuit $\bigl(X_{N}\CNOTladderN{N}\bigr)^{\,2^{\lceil\log_{2}(N+1)\rceil}}$ occur in even pairs and cancel. 
\begin{align}
\bigl(X_N\CNOTladderN{N}\bigr)^{\,2^{\lceil\log_{2}(N+1)\rceil}}
= \bigl(\CNOTladderN{N}\bigr)^{\,2^{\lceil\log_{2}(N+1)\rceil}}
=E^{(N)}.
\end{align}
Therefore, following equation also holds,
\begin{align}
\bigl(\CNOTladderN{N+1}\bigr)^{\,2^{\lceil\log_{2}(N+1)\rceil}}
       =E^{(N+1)}.
\end{align}

Under the induction hypothesis for $N$, we have shown that the statement also holds for $N+1$.  Therefore, by mathematical induction, Proposition~\ref{prop:order-of-CNOT-ladder} is proved.

\end{proof}

\end{document}